\documentclass[12pt,a4paper]{amsart}
\usepackage{latexsym}
\usepackage{graphics} 
\usepackage{epsfig} 
\usepackage{amsmath} 
\usepackage{amsfonts, amssymb} 
\usepackage{amsthm,enumerate}
\usepackage{color}
\usepackage{mathrsfs}
\usepackage{amscd}      
\usepackage{ulem}
\usepackage{euscript}
\usepackage{colortbl}
\usepackage{arydshln}
\usepackage{shuffle}
\usepackage[all]{xy}

\usepackage[
	a4paper,
	top=25mm,
	bottom=20mm,
	left=25mm,
	right=25mm,
]{geometry}

\theoremstyle{definition}
\newtheorem{Def}{Definition}[section]

\newtheorem{Eg}{Example}[section]
\newtheorem{Rm}{Remark}[section]

\theoremstyle{plain}
\newtheorem{Prop}[Def]{Proposition}

\newtheorem{Thm}[Def]{Theorem}
\newtheorem{Cor}[Def]{Corollary}

\numberwithin{equation}{section}

\newcommand{\authorfootnotesA}{\renewcommand\thefootnote{$\flat$}}%
\newcommand{\authorfootnotesB}{\renewcommand\thefootnote{$\sharp$}}%
\newcommand{\authorfootnotesC}{\renewcommand\thefootnote{$\diamond$}}%
\newcommand{\authorfootnotesD}{\renewcommand\thefootnote{$\S$}}%

\begin{document}

\begin{center}
	\LARGE 
	Modulus of Continuity of Controlled Loewner-Kufarev Equations
	and Random Matrices
	\par \bigskip

	\normalsize
	\authorfootnotesA
	Takafumi Amaba\footnote{
	This work was supported by JSPS KAKENHI Grant Number 15K17562.
	}\textsuperscript{1}
	and
	\authorfootnotesB
	Roland Friedrich\footnote{Partially supported by the ERC advanced grant ``Noncommutative  
        distributions in free probability".}\textsuperscript{2}
	\authorfootnotesC
	\authorfootnotesD

	\textsuperscript{1}Fukuoka University,
	8-19-1 Nanakuma, J\^onan-ku, Fukuoka, 814-0180, Japan.\par \bigskip

	\textsuperscript{2}Saarland University,
	Faculty of Mathematics,
	D-66123 Saarbr{\"u}cken,
	Germany\par \bigskip

	
	\email{(T.~Amaba) fmamaba@fukuoka-u.ac.jp}
	\email{(R.~Friedrich) friedrich@math.uni-sb.de}

	\today
\end{center}

\begin{quote}{\small {\bf Abstract.}
First we introduce the two tau-functions which appeared either as the $\tau$-function of the integrable hierarchy governing the Riemann mapping of Jordan curves or in conformal field theory and the universal Grassmannian. Then we discuss various aspects of their interrelation. Subsequently, we establish a novel connection between free probability,
growth models
and integrable systems, in particular for second order freeness, and summarise it in a dictionary. This extends the previous link between conformal maps and large $N$-matrix integrals to (higher) order free probability.
Within this context of dynamically evolving contours,
we determine a class of driving functions for
controlled Loewner-Kufarev equations,
which enables us to give a continuity estimate
for the solution to such equations
when embedded into the Segal-Wilson Grassmannian.
}
\end{quote}

\renewcommand{\thefootnote}{\fnsymbol{footnote}}

\footnote[0]{ 
2010 \textit{Mathematics Subject Classification}.
Primary 93C20; 
Secondary
30F10, 
35C10, 
58J65  
}

\footnote[0]{ 
\textit{Key words and phrases}.
Loewner-Kufarev equation,
Grassmannian,
Witt algebra,
Faber polynomial,
Grunsky coefficient,
Signature,
Control function
}

\section{Introduction}

The class of univalent functions is an extraordinarily rich
mathematical object within the field of complex variables,
with deep and surprising connections with, e.g.
conformal field theory (CFT),
random matrix theory and integrable systems,
cf.~\cite{Ju1993,KNTY,KKMWZ2000,KY1988,MWZ2000,Ta01,Te03,WZ00},
just to name the most important ones in our context.
So, for
$
\mathbb{D} := \{ z \in \mathbb{C} | |z|<1 \}
$,
the open unit disc, with boundary the unit circle,
i.e.
$S^1 = \partial \mathbb{D}$,
let
$$
\mathcal{S}
:=
\{ f: \mathbb{D} \rightarrow \mathbb{C}~|~\text{
	univalent and regular up to $S^1$,
	$f(0)=0$ and $f'(0)=1$
	}
\},
$$
be the class of schlicht functions.
Then, for every $f \in \mathcal{S}$,
$D := f( \mathbb{D} )$
is a simply connected domain containing the origin,
with boundary $C := \partial D$, a Jordan contour.
If $\hat{\mathbb{C}} := \mathbb{C} \cup \{ \infty \}$
is the Riemann sphere,
let
$D^c := \hat{\mathbb{C}} \setminus D$
be the complement of $D$ in the extended complex plane.

The set $\mathcal{C}$ of all such Jordan contours encircling
the origin forms an infinite dimensional manifold~\cite{Ta01,KMZ05}.

It has been shown by A. Kirillov and D. Juriev~\cite{KY1988}
that there exists a canonical bijection
$$
\mathcal{S} \cong \operatorname{Diff}_+ (S^1) / S^1
$$
which endows
$\operatorname{Diff}_+(S^1)/S^1$
with the structure of an infinite-dimensional complex manifold.

Geometrically,
$\pi : \mathcal{C} \rightarrow \mathcal{S}$
is a fibre bundle, with fibre $\mathbb{R}_+^*$,
which is a consequence of the Riemann mapping theorem.
There exist two continua of global sections
$
\sigma_{r_i} : \mathcal{S} \rightarrow \mathcal{C}
$,
$r_i>0$, $i=1,2$,
such that the leaves
$
\mathcal{C}_{r_i} := \sigma_{r_i} (\mathcal{S})
$
stratify $\mathcal{C}$, i.e.
$
\mathcal{C} = \biguplus_{r_{1}>0} \mathcal{C}_{r_{1}}
= \biguplus_{r_{2}>0} \widetilde{\mathcal{C}}_{r_{2}}
$,
either according to the conformal radius $r_{1} > 0$,
or alternatively,
the interior area $r_{2}>0$,
as in \cite{Ta01}.

Krichever, Marshakov, Mineev-Weintstein, Wiegmann and Zabrodin~\cite{WZ00,MWZ02,KMZ05},
motivated by the
Hele-Shaw problem, cf.~\cite{MWZ2000} and the monograph by Gustafsson, Teodorescu and Vasil'ev~\cite{GTV2014},
introduced a new set of co-ordinates,
the so-called harmonic moments of the interior and exterior domain,
with respect to a family of harmonic functions.
Namely, for
$\{ z^{-k} \}_{ k \in \mathbb{Z}_{\geqslant 0} }$
the interior harmonic moments are given by:
\begin{eqnarray}
t_k
:=
- \frac{1}{\pi k}
\int_{D^c} z^{-k} \mathrm{d}^2 z
=
\quad
\oint_{\partial D} z^{-k} \bar{z} \mathrm{d}z,
\end{eqnarray}
where the second equality is a consequence of Stokes' Theorem.
Further,
$$
t_0 := \frac{1}{\pi} \int_D \mathrm{d}^2 z.
$$
is the area with respect to Lebesgue measure.
The set
$$
\mathbf{t}_{\pm}
:=
( t_0, t_1, \bar{t}_1 , t_2, \bar{t}_2 , t_3, \bar{t}_3 , \dots )
\in \mathbb{R}_+ \times \mathbb{C}^{\mathbb{N}}, 
$$ 
with $\bar{t}_k$ denoting the complex conjugate of $t_k$,
are local co-ordinates on the manifold
(moduli space)
of smooth closed contours $\mathcal{C}$, cf.~\cite{Ta01}.

The other set of (natural) co-ordinates is given by the coefficients
of the normalised Riemann mapping.
So, we have two different sets of co-ordinates for $\mathcal{C}$,
as shown below:
\begin{equation*}
\begin{xy}
\xymatrix{
	&
	\mathbb{R}^*_+ \times \mathbb{C}^{\mathbb{N}}
	& \\
	\operatorname{Aut}(\mathcal{O}) | \operatorname{Der}_+(\mathcal{O})
	\ar[ur]^-{
		\begin{array}{c}
		\text{{\tiny The first co-ordinate}} \vspace{-2mm}\\
		\text{{\tiny respects the}} \vspace{-2mm}\\
		\text{{\tiny conformal radius}}
		\end{array}
	}
	&
	&
	\{\mathbf{t}_{\pm}\}
	|
	\{ \partial_{t_0}, \partial_{t_k}, \partial_{\bar{t_k}} \}_{ k\in\mathbb{N} }
	\ar[ul]_-{
		\begin{array}{c}
		\text{{\tiny The first co-ordinate}} \vspace{-2mm}\\
		\text{{\tiny respects the}} \vspace{-2mm}\\
		\text{{\tiny interior area}}
		\end{array}
	}
}
\end{xy}
\end{equation*}
Hence, the tangent space to $\mathcal{C}$ permits also two descriptions,
namely, as~\cite{FBZ,KNTY,KMZ05,Ta01}
$$
\operatorname{Der}_0 ( \mathcal{O} )
:=
z \mathbb{C} [\![ z ]\!]
\partial_z
\quad
\text{and}
\quad
\{ \partial_{t_0}, \partial_{t_k}, \partial_{\bar{t_k}} \}_{k\in\mathbb{N}}
$$
Define
$
\ell_n
:=
-z^{n+1}
\frac{d}{dz}
$,
$n \in \mathbb{N}$,
which span the positive part of the Witt algebra,
i.e. for $n,m\in\mathbb{Z}$
$$
[ \ell_m, \ell_n ] = (m-n) \ell_{m+n}.
$$
The $\partial_{t_k}$ can be determined either by
specific boundary variations of the domain,
which do only change one harmonic moment at the time,
cf.~\cite[Formula~(2.11)]{KMZ05},
or in terms of the Faber polynomials~\cite{Ta01}.
Combining results in \cite{KMZ05} with
our considerations,
the relation for the different tangent vectors is given by

\begin{Prop} 
The vector fields
$\{ \partial_{t_{k}} \}_{k \geqslant 1}$
on $\mathcal{C}$
and the operators
$\{ \ell_{k} \}_{k \geqslant 1}$ are related by
\begin{equation*}
\partial_{t_{k}} t_{l}
=
\frac{1}{ k \pi }
\oint_{ \partial D^{c} }
\xi^{-l}
\delta n ( \xi )
\vert \mathrm{d}\xi \vert
=
\delta_{kl},
\end{equation*}
where
\begin{equation*}
\delta n ( \xi )
:=
\partial_{n}
\frac{-1}{2 \pi i}
\oint_{\infty} ( \ell_{k} G_{0} )( z, \xi ) \frac{ \mathrm{d}z }{z},
\end{equation*}
and $\partial_{n}$
is the normal derivative on the boundary $\partial D^{c}$
with respect to $\xi \in \partial D^{c}$,
and
$G_{0}( x, \xi )$
is the Dirichlet Green function
associated to the Dirichlet problem in $D^{c}$.
\end{Prop} 

Krichever, Marshakov, Mineev-Weinstein, Wiegmann and Zabrodin~\cite{KKMWZ2000,MWZ02,WZ00},
in different constellations,
defined the logarithm of a
{\it $\tau$-function}
which
for
a contour $C = \partial D$ is given
by
\cite{KKMWZ2000,Ta01}
\begin{equation*}
\ln ( \tau ):
\mathcal{C} \ni C
\mapsto
- \frac{ 1 }{ \pi^{2} }
\int_{D} \int_{D}
\ln \Big\vert \frac{1}{z} - \frac{1}{w} \Big\vert
\mathrm{d}^{2}z \mathrm{d}^{2}w
\in \mathbb{R}.
\end{equation*}

The $\tau$-function connects complex analysis with the dispersionless
hierarchies and integrable systems~\cite{KKMWZ2000}

A key result,
which expresses the Riemann mapping in terms of the $\tau$-function,
is the following 
\begin{Thm}[\cite{KKMWZ2000,Ta01}] 
Let
$
g :
\hat{\mathbb{C}} \setminus D
\rightarrow
\hat{\mathbb{C}} \setminus \mathbb{D}
$
be the conformal map,
normalised by $g( \infty ) = \infty$
and
$g^{\prime} ( \infty ) > 0$.
Then the following formula holds:
$$
\ln ( g(z) )
=
\ln(z)
-
\frac{1}{2}
\frac{\partial^2\ln(\tau)}{\partial t^2_0}
-
\sum_{k=1}^{\infty} \frac{z^{-k}}{k}
\frac{\partial^2 \ln (\tau)}{\partial t_0\partial t_n}.
$$
\end{Thm} 

This formula would be a key to describe the solution to
the conformal welding problem
(c.f.~\cite{Te08})
associated to Malliavin's canonic diffusion
\cite{Ma99}
within the framework of Loewner-Kufarev equation,
which would be a future work.

Another interpretation of the
$\tau$-function,
given by Takhtajan~\cite[Corollary 3.10]{Ta01},
is that it is a K\"{a}hler potential of a Hermitian metric
on
$\widetilde{\mathcal{C}_a}$, $a>0$.

Kirillov and Juriev~\cite{KY1988},
defined a two-parameter family $(h,c)$ of K\"{a}hler potentials $K_{h,c}$ on the determinant line bundle
$\operatorname{Det}^*$ over the (Sato)-Segal-Wilson Grassmannian, where $h$ is the highest-weight and $c$ the central charge of the CFT. For $h=0$ and $c=1$, i.e. the free Boson, one has~\cite{KY1988, Ju1993} for the metric
$$
g_{0,1}(f)
=
\mathrm{e}^{{-K_{0, 1}(f)}}
\mathrm{d} \lambda
\mathrm{d} \bar{\lambda}
$$
where $\lambda$ is the co-ordinate in the fibre over the schlicht function $f$.
In terms of the Grunsky matrix $Z_f$, associated to an $f \in \mathcal{S}$, we obtain 
the relation between the K\"{a}hler potential, for $h=0, c=1$, and the $\tau$-function, as
$$
\ln (\tau)(f)
\sim
\ln
\det (1-Z_f \bar{Z_f}).
$$

The following diagram summarises
our
discussion so-far:
\begin{equation*}
\xymatrix@=20pt{
&
\mathbb{R}
&
&
& \\
\mathrm{Aut} ( \mathcal{O} )
\ar[d]_-{}
&
\mathcal{C}
\ar[r]_-{}
\ar[d]_-{\pi}
\ar[l]
\ar@{->}[u]^-{\ln (\tau)} 
&
\mathrm{Det}^{*}\vert_{M}
\ar@{^{(}->}[rr]
\ar@{->}[ul]_-{
	\begin{array}{c}
	\text{{\tiny K\"{a}hler potential}} \vspace{-2mm}\\
	\text{{\tiny $\ln \det ( 1 - Z_{f} \bar{Z}_{f} )$}}
	\end{array}
}
\ar[d]^-{}
&
&
\mathrm{Det}^{*}
\ar[d]^-{}
& \\
\mathrm{Aut}_{+} ( \mathcal{O} )
&
\mathcal{S}
\ar[r]^-{Z}_-{\text{Krichever}}
\ar[l]
&
M
\ar@{}[r]|-*{:=}
&
\{ Z_{f} : f \in \mathcal{S} \}
\ar@{}[r]|-*{\subset}
&
\mathrm{Gr}_{\text{SW}}
\ar[r]^-{\tau_{\text{SW}}}
&
\mathbb{C} [\![ t_1 , t_2 , t_3 , \dots ]\!] = \mathcal{H}^*_{0}
}
\end{equation*}
where $Z$ is the Grunsky matrix,
cf.~\cite{KY1988} and $\mathcal{H}^*_0$
the charge $0$ sector of the boson Fock space~\cite[p.279]{KNTY}.
For the second square from the left, we should note that
the
Krichever mapping does not distinguish
$\mathcal{S}$ and $\mathcal{C}$
algebro-geometrically.
Namely, the Krichever embedding of a Riemann mapping uses
only the negative part of the Grunsky coefficients $b_{-m,-n}$, $m,n =1, 2,\dots$
but not $b_{0,0}$. One finds from the defining equation of the Grunsky coefficients that
$b_{0,0}$ is the only entry of the Grunsky matrix which depends on the conformal radius.
Consequently, Krichver's embedding forgets about the conformal radius.
But in order to keep track of the modulus of the derivative of the 
normalised Riemann mapping,
we put that information into the determinant line bundle.
This is the mapping $\mathcal{C} \to \mathrm{Det}^{*}\vert_{M}$.

The structure of  the rest of the paper is as follows:
In Section~\ref{Sec_LFP} we establish a relation between
the theory of Laplacian growth models, and their integrable structure,
with a class of random matrices and second order free probability.
We succinctly summarise it in a dictionary.
In Section~\ref{Sec_CLKE} we consider controlled Loewner-Kufarev equations
and recall the necessary facts. Then,
in Section~\ref{Sec_CLK},
we give several estimates for the Grunsky coefficients
associated to solutions to a controlled Loewner-Kufarev equation.
Proofs of several estimates which need results from
\cite{AmbFr}
are relegated to Appendix~\ref{Appdx}.
Finally, we prove Theorem~\ref{modulus-LK} in Section~\ref{Sec_LK/Gr}.

\section{Integrability and Higher Order Free Probability}
\label{Sec_LFP}
Another motivation in the works of Marshakov et al.
was the close connection random matrix theory red has with
(Laplacian)
growth models and integrable hierarchies.
Takebe, Teo and Marshakov discussed the geometric meaning
of the eigenvalue distribution in the large $N$ limit of normal random matrices
in conjunction with the one variable reduction via the Loewner equation~\cite{TTZ}.  In~\cite{AmbFr} we established and discussed a relation
between CFT and free probability theory. 
Here we briefly present a novel connection between integrable hierarchies,
large $N$ limits of Gaussian random matrices and 
second (higher) order
free probability~\cite{CMSS}.
First, consider, cf.~\cite[p.~42]{Ta01},
$$
\langle\!\langle j(z) j(w) \rangle\!\rangle
$$
the normalised current two-point functions for free Bososn on $\hat{\mathbb{C}}\setminus D$, and the analogous correlation function with the Dirichlet boundary conditions, i.e.
$$
\langle\!\langle j(z) j(w) \rangle\!\rangle_{\text{DBC}}.
$$
Further, let, cf.~\cite[p.~11]{CMSS},
$$
G(z,w):=\frac{M(\frac{1}{z},\frac{1}{w})}{zw},
$$ 
be the second order Cauchy transform, and $M(\frac{1}{z},\frac{1}{w})$ the second order moments. We obtain

\begin{Thm} 
Assume that
the
second order free cumulants
$R(z,w)$ vanish, i.e. we have an integrability /zero-curvature condition,
such as for e.g. the Gaussian and Wishart random matrices.
Then the tensor corresponding to the second order Cauchy transform
$
G(z,w) \mathrm{d}
z
\otimes
\mathrm{d}
w
$
is given by the Ward identity 
\begin{eqnarray}
G(z,w) \mathrm{d}z \otimes \mathrm{d}w
&=&
\langle\!\langle j(z) j(w) \rangle\!\rangle
-
\langle\!\langle j(z) j(w) \rangle\!\rangle_{\mathrm{DBC}} \\
&=&
\left(
	\frac{ G^{\prime}(z) G^{\prime}(w) }{ ( G(z) - G(w) )^2 }
	-
	\frac{1}{(z-w)^2}
\right)
\mathrm{d}z \otimes \mathrm{d}w \\
&=&\sum_{m,n=1}^{\infty}z^{-m-1}w^{-n-1}\frac{\partial^2\ln(\tau)}{\partial t_m\partial t_n}\mathrm{d}z\otimes \mathrm{d}w
\end{eqnarray}
where $G(z)$ is the first order Cauchy transform.
\end{Thm} 

This suggests us to make a dictionary
translating languages from a integrable system and free probability.
Table~\ref{dict} is an attempt to list objects in these fields
sharing the same algebraic relations.
\begin{table}[h]
\centering
\caption{}
\label{dict} 
\begin{tabular}{|p{7cm}|p{7cm}|}
\hline
\begin{center}
Laplacian Growth Model \\
(dToda hierachy)
\end{center}
&
\begin{center}
Free Probability Theory
\end{center}
\\ \hline
\begin{minipage}{7truecm}
\centering
\phantom{a}\vspace{-2mm}
Exterior of the domain $D$,
$$D^{c} \vspace{2mm}$$
\end{minipage}
&
\begin{minipage}{7truecm}
\centering
$$\text{Spectral measure}^c$$
\end{minipage}
\\ \hdashline
\begin{minipage}{7truecm}
\centering
\phantom{a}\vspace{-2mm}
Harmonic moments of $D$,
$$
v_{n}
=
\frac{1}{2\pi i}
\int_{\partial D}
z^{n} \overline{z} \mathrm{d}z
=
\frac{\partial \log \tau}{\partial t_{n}},
$$
$$
v_{0}
=
\frac{2}{\pi}
\int_{D}
( \log \vert z \vert )
\overline{z} \mathrm{d}z
=
\frac{\partial \log \tau}{\partial t_{0}}
$$
\cite[Corollary~3.3]{Ta01}
\vspace{2mm}
\end{minipage}
&
\begin{minipage}{7truecm}
\centering
$$?$$
\end{minipage}
\\ \hdashline
\begin{minipage}{7truecm}
\centering
\phantom{a}\vspace{-2mm}
Derivative of harmonic moments,
$$
\frac{\partial v_{m}}{\partial t_{n}}
\quad
\text{for $m,n \geqslant 1$}
\vspace{2mm}
$$
\end{minipage}
&
\begin{minipage}{7truecm}
\centering
The second order limit moments,
$$
\alpha_{m,n}
\quad
\text{for $m,n \geqslant 1$}
$$
\end{minipage}
\\ \hdashline
\begin{minipage}{7truecm}
\centering
Riemann mapping,
$$
G: D^{c} \to \mathbb{D}^{c}
$$
\end{minipage}
&
\begin{minipage}{7truecm}
\centering
\phantom{a}\vspace{-2mm}
The first order Cauchy transform,
$$
G(x) = \frac{M(\frac{1}{x})}{x}
\vspace{2mm}
$$
\end{minipage}
\\ \hdashline
\begin{minipage}{7truecm}
\centering
\phantom{a}\vspace{-2mm}
Normalised reduced $2$-point current correlation function of free bosons on $\mathbb{P}^{1}$
parametrised by $C = \partial D$,
$$
G(z,w)
=
\langle\!\langle \jmath (z) \jmath (w) \rangle\!\rangle
-
\langle\!\langle \jmath (z) \jmath (w) \rangle\!\rangle_{\mathrm{DBC}}
$$
$$
\phantom{G(z,w)}
=
\frac{\partial^{2} G_{0}(z,w)}{\partial z \partial w}
-
\frac{\partial^{2}G_{\mathrm{DBC}} (z,w)}{\partial z \partial w} 
$$
$$
\phantom{G(z,w)}
=
\sum_{m,n=1}^{\infty}
z^{-m-1} w^{-n-1}
\frac{ \partial^{2} \ln ( \tau ) }{ \partial t_{m} \partial t_{n} },
$$
where $G_{0} (z,w)$ is the Green function
associated with the $\overline{\partial}$-Laplacian
$$
\triangle_{0} = - \frac{1}{2} * \overline{\partial} * \overline{\partial}
$$
and $G_{\mathrm{DBC}}$ is the Green function
associated to $\triangle_{0}$ on $D^{c}$ \\
\cite[Theorem~3.9]{Ta01}
\vspace{2mm}
\end{minipage}
&
\begin{minipage}{7truecm}
\centering
The second order Cauchy transform,
$$G(x,y) = \frac{M(\frac{1}{x}, \frac{1}{y})}{xy}$$
\end{minipage}
\\ \hdashline
\begin{minipage}{7truecm}
\centering
\phantom{a}\vspace{-2mm}
Normalised $1$-point current correlation function
of free bosons on $\mathbb{CP}^{1}$ parametrised by $C = \partial D$,
$$
\langle\!\langle
\jmath (z)
\rangle\!\rangle
=
- \sum_{n=1}^{\infty} z^{-n-1} \frac{ \partial \log \tau }{ \partial t_{n} }
$$
\text{{\small \cite[Theorem~3.1]{Ta01}}}
\vspace{2mm}
\end{minipage}
&
\begin{minipage}{7truecm}
\centering
$$?$$
\end{minipage}
\\ \hline
\end{tabular}
\end{table}

From the above it follows now,
that general higher order free (local) cumulants
are given by Ward identities of $n$-point functions
of the twisted Boson field over arbitrary Riemann surfaces.

\section{Controlled Loewner-Kufarev Equations}
\label{Sec_CLKE}

The connection with the Loewner equation,
the class of schlicht functions and integrable systems
was established by
Takebe, Teo and Zabrodin~\cite{TTZ}.
They showed that both the chordal and the radial Loewner equation
give consistency conditions of such integrable hierarchies.
A particularly important class of such consistency conditions
can be obtained from specific control functions.

In the previous paper \cite{AmbFr},
the authors
introduced the  notion of a solution to the
{\it controlled Loewner-Kufarev equation}
(see \cite[Definition~2.1]{AmbFr})
\begin{equation}
\label{controlled-LK} 
\mathrm{d} f_{t}(z)
=
z f_{t}^{\prime} (z)
\{
	\mathrm{d} x_{0} (t)
	+
	\mathrm{d} \xi ( \mathbf{x}, z )_{t}
\},
\quad
f_{0} (z) \equiv z \in \mathbb{D}
\end{equation}
where
$\mathbb{D}= \{ \vert z \vert < 1 \}$
is the unit disc in the complex plane, 
$
x_{0} : [0,T] \to \mathbb{R}
$,
$$
x_{1}, x_{2}, \cdots : [0,T] \to \mathbb{C}
$$
are given continuous functions of bounded variation,
called the
{\it driving functions}.
We define
$$
\xi ( \mathbf{x}, z )_{t}
:=
\sum_{n=1}^{\infty} x_{n}(t) z^{n}.
$$

In the current paper,
we determine a class of  driving functions
for which we establish the continuity
of the solution,
as a curve embedded in the
(Sato)-Segal-Wilson Grassmannian,
with respect to time.
For this reason, we introduce the following class of 
controlled Loewner-Kufarev equations.

Let us first recall from the monograph by Lyons and Qian~\cite[Section~2.2]{LyQi}
a few basic notions. 
For a fixed $T \geqslant 0$,
let $\Delta_T:\{(s,t):0\leqslant s \leqslant t\leqslant T\}$ be the two-simplex.

\begin{Def}[{(see \cite[Section~2.2]{LyQi})}] 
\label{Def:control} 
A continuous function
$$
\omega :
\{ (s,t) : 0 \leqslant s \leqslant t < +\infty \}
\to
\mathbb{R}_+,
$$
is called a
{\it control function}
if it satisfies
super-additivity:
$$\omega (s,u) + \omega (u,t) \leqslant \omega (s,t),$$
for all $0 \leqslant s \leqslant u \leqslant t$,
and vanishes on the diagonal, i.e. $\omega(t,t)=0$, for all $t\in[0,T]$.
\end{Def} 

Now, let $V$ be a Banach space and $X: [0,T] \rightarrow V$,
$T \geqslant 0$, be a path such that
$$
\vert X_t - X_s \vert \leqslant \omega (s,t),
\quad
0 \leqslant s \leqslant t \leqslant T,
$$
for a control function
$\omega : \Delta_T \to \mathbb{R}_+$.
Then $X$ is called a Lipschitz path controlled by $\omega$.

We have the following

\begin{Def} 
\label{LKw/omega} 
Let $\omega$ be a control function.
The driven
Loewner-Kufarev equation\footnote{
Here we avoid the conflicting use of the word
``controlled",
which has two meanings, cf.~\cite{LyQi} p.~16.
}
(\ref{controlled-LK})
is {\it controlled by $\omega$} if
for any $n \in \mathbb{N}$,
$p=1, \cdots , n$
and
$i_{1}, \cdots , i_{p} \in \mathbb{N}$
with $i_{1} + \cdots + i_{p} = n$,
we have
\begin{equation*}
\begin{split}
&
\Big\vert
\mathrm{e}^{ n x_{0} (t) }
\int_{
	0 \leqslant u_{1} < \cdots <
	u_{p}
	\leqslant t
}
\mathrm{e}^{-i_{1} x_{0} (u_{1})}
\mathrm{d} x_{i_{1}} (u_{1})
\cdots
\mathrm{e}^{
	- i_{p} x_{0}
	( u_{p} )
	}
\mathrm{d} x_{i_{p}} (u_{p})
\Big\vert
\leqslant
\frac{
	\omega (0,t)^{n}
}{ n! },
\end{split}
\end{equation*}
and
\begin{equation*}
\begin{split}
&
\Big\vert
\mathrm{e}^{ n x_{0} (t) }
\int_{
	0 \leqslant u_{1} < \cdots <
	u_{p}
	\leqslant t
}
\mathrm{e}^{-i_{1} x_{0} (u_{1})}
\mathrm{d} x_{i_{1}} (u_{1})
\cdots
\mathrm{e}^{
	- i_{p}
	x_{0}
	(
	u_{p}
	)
	}
\mathrm{d} x_{i_{p}} (u_{p}) \\
&\hspace{10mm}-
\mathrm{e}^{ n x_{0} (s) }
\int_{
	0 \leqslant u_{1} < \cdots <
	u_{p}
	\leqslant s
}
\mathrm{e}^{-i_{1} x_{0} (u_{1})}
\mathrm{d} x_{i_{1}} (u_{1})
\cdots
\mathrm{e}^{
	- i_{p} x_{0}
	(u_{p})
	}
\mathrm{d} x_{i_{p}} (u_{p})
\Big\vert \\
&\hspace{20mm}
\leqslant
\omega (s,t)
\frac{
	\omega (0,T)^{n-1}
}{ (n-1)! },
\end{split}
\end{equation*}
for any $0 \leqslant s \leqslant t \leqslant T$.

Henceforth, we will refer to equation
(\ref{controlled-LK})
as the {\it Loewner-Kufarev equation controlled by $\omega$},
or the {\it $\omega$-controlled Loewner-Kufarev equation.}
\end{Def} 

A natural question to be asked is,
how a control function as 
driving function determines a control function for
(\ref{controlled-LK}).
We will give one of the answers in Corollary~\ref{Drive>>CLK}.

Let $H = L^{2} ( S^{1},\mathbb{C} )$ be the
Hilbert space of all square-integrable complex-valued functions
on the unit circle $S^{1}$,
and we denote by $\mathrm{Gr} := \mathrm{Gr}(H)$
the Segal-Wilson Grassmannian
(see \cite[Definition~3.1]{AmbFr} or \cite[Section~2]{SW}).
Any bounded univalent function
$f : \mathbb{D} \to \mathbb{C}$
with $f(0) = 0$
and $\partial f( \mathbb{D} )$ being a Jordan curve,
is embedded into $\mathrm{Gr}$ via
\begin{equation*}
f
\mapsto
W_{f}
:=
\overline{
\mathrm{span}
\big(
	\{ 1 \}
	\cup
	\{ Q_{n} \circ f \circ (1/z)\vert_{S^{1}} \}_{n \geqslant 1}
\big)
}^{H}
\in \mathrm{Gr}
\end{equation*}
(see \cite[Sections~3.2 and 3.3]{AmbFr}),
where $Q_{n}$ is the $n$-th Faber polynomial associated to $f$.

Note that $f$ extends to a continuous function on
$\overline{\mathbb{D}}$
by Caratheodory's Extension Theorem for holomorphic functions.

Let $H^{1/2} = H^{1/2}(S^{1})$ be the Sobolev space on $S^{1}$
endowed with the inner product given by
$
\langle
	h, g
\rangle_{H^{1/2}}
=
h_{0} \overline{g}_{0}
+
\sum_{n \in \mathbb{Z}}
\vert n \vert
h_{n} \overline{g}_{n}
$
for
$
h = \sum_{n \in \mathbb{Z}} h_{n} z^{n},
g = \sum_{n \in \mathbb{Z}} g_{n} z^{n} \in H^{1/2}
$.
Assume that $f$ extends to a holomorphic function
on an open neighbourhood of $\overline{\mathbb{D}}$.
Then
$
\mathrm{span}
\big(
	\{ 1 \}
	\cup
	\{ Q_{n} \circ f \circ (1/z)\vert_{S^{1}} \}_{n \geqslant 1}
\big)
\subset
H^{1/2}
$
and we consider the orthogonal projection
\begin{equation*}
\EuScript{P}_{f}: H^{1/2} \to W_{f}^{1/2},
\quad
\text{where
$
W_{f}^{1/2}
:=
\overline{
\mathrm{span}
\big(
	\{ 1 \}
	\cup
	\{ Q_{n} \circ f \circ (1/z)\vert_{S^{1}} \}_{n \geqslant 1}
\big)
}^{H^{1/2}}
$}
\end{equation*}
rather than the orthogonal projection
$
H \to W_{f}
$.

Let $\omega$ be a control function and
let $\{ f_{t} \}_{0 \leqslant t \leqslant T}$
be a
univalent
solution to the Loewner-Kufarev equation
controlled by $\omega$.
Suppose that each $f_{t}$ extends to a holomorphic function
on an open neighbourhood of $\overline{\mathbb{D}}$.
With
these
assumptions, our main result is the following

\begin{Thm} 
\label{modulus-LK} 

Suppose that $\omega (0,T) < \frac{1}{8}$.
Then there exists a constant
$c = c(T) > 0$
such that
\begin{equation*}
\begin{split}
\Vert
	\EuScript{P}_{f_{t}}
	-
	\EuScript{P}_{f_{s}}
\Vert_{\mathrm{op}}
\leqslant
c \hspace{0.5mm} \omega (s,t)
\end{split}
\end{equation*}
for every $0 \leqslant s < t \leqslant T$,
where
$\Vert \bullet \Vert_{\mathrm{op}}$
is the operator norm.
\end{Thm} 

Thus we obtained a continuity result with respect to the time-variable
of the solution embedded into the Grassmannian
in which the modulus of continuity is measured by
the control function $\omega$.
Let us point out that the assumption of the existence of an
analytic continuations of $f_{t}$'s across $S^{1}$
is extrinsic and should be discussed in detail in the future.

\section{Auxiliary Estimates Along Controlled Loewner-Kufarev Equations}
\label{Sec_CLK} 

\subsection{Controlling Loewner-Kufarev equation by its driving function}

We shall begin with a prominent example of
a control function as follows.

\begin{Eg} 
If
$y : [0,+\infty ) \to \mathbb{C}$,
is continuous and of bounded variation,
then we have
\begin{equation*}
\begin{split}
\Vert
	y
\Vert_{\text{$1$-var}(s,t)}
:=
\sup_{
	\substack{
		n \in \mathbb{N} ; \\
		s = u_{0} < u_{1} < \cdots < u_{n-1} < u_{n}=t
	}
}
\sum_{i=1}^{n}
\vert
y_{u_{i}} - y_{u_{i-1}}
\vert
< +\infty
\end{split}
\end{equation*}
for every $0 \leqslant s \leqslant t$.
Then
$
\omega (s,t)
:=
\Vert y \Vert_{\text{$1$-Var}(s,t)}
$
defines a control function.
\end{Eg} 

\begin{Def} 
Let
$
\omega
$
be a control function.
We say
that
a continuous function
$y : [0,T] \to \mathbb{C}$
of bounded variation,
{\it
is controlled by $\omega$
}
if
$
\Vert y \Vert_{\text{$1$-var}(s,t)}
\leqslant \omega (s,t)
$
for every $0 \leqslant s \leqslant t \leqslant T$.
\end{Def} 

As is well known,
introducing control functions makes our calculations stable
as follows.

\begin{Eg} 
\label{iterated:controlled} 
Let $n\in\mathbb{N}$ and
$y_{1}, \cdots , y_{n} : [0,+\infty ) \to \mathbb{C}$
be continuous and controlled by a control function $\omega$.
Then we have
\begin{equation*}
\begin{split}
&
\sup_{
	\substack{
		m \in \mathbb{N} ; \\
		s = r_{0} < r_{1} < \\
		\hspace{5mm}\cdots < r_{m-1} < r_{m}=t
	}
}
\sum_{i=1}^{m}
\big\vert
\int_{
	r_{i-1} \leqslant u_{1} < \cdots < u_{n} \leqslant r_{i}
}
\mathrm{d} y_{1} (u_{1})
\cdots
\mathrm{d} y_{n} (u_{n})
\big\vert
\leqslant 
\frac{ \omega (s,t)^{n} }{ n! }
\end{split}
\end{equation*}
for every $0 \leqslant s \leqslant t$.

In fact, we shall prove this by induction on $n$.
The case for $n=1$ is clear by definition.
Consider the case for $n-1$.
Putting $\omega_{s}(t) := \omega (s,t)$,
we find that the total variation measure
$\vert \mathrm{d} y_{n} \vert$ on $[s, +\infty )$
is smaller than the Lebesgue-Stieltjes measure
$\mathrm{d} \omega_{s}$
associated with $\omega_{s}$
on $[s, +\infty )$,
in the sense of
$
\int_{B} \vert \mathrm{d} y_{n} \vert
\leqslant
\int_{B} \mathrm{d} \omega_{s}
$
for any Borel set $B \subset [s,+\infty )$.
Therefore we have
\begin{equation*}
\begin{split}
\big\vert
\int_{
	r_{i-1} \leqslant u_{1} < \cdots < u_{n} \leqslant r_{i}
}
\mathrm{d} y_{1} (u_{1})
\cdots
\mathrm{d} y_{n} (u_{n})
\big\vert
&\leqslant 
\int_{r_{i-1}}^{r_{i}}
\Big\vert
\int_{
	r_{i-1} \leqslant u_{1} < \cdots < u_{n-1} \leqslant u_{n}
}
\mathrm{d} y_{1} (u_{1})
\cdots
\mathrm{d} y_{n-1} (u_{n-1})
\Big\vert
\vert \mathrm{d} y_{n} (u_{n}) \vert \\
&\leqslant 
\int_{r_{i-1}}^{r_{i}}
\frac{ \{ \omega_{r_{i-1}} (u) \}^{n-1} }{ (n-1)! }
\mathrm{d} \omega_{r_{i-1}} (u)
=
\frac{ \omega ( r_{i-1}, r_{i} )^{n} }{ n! }.
\end{split}
\end{equation*}
Since the control function is nonnegative
and super-additive,
it holds that
$$
\sum_{i=1}^{m} \omega ( r_{i-1}, r_{i} )^{n}
\leqslant
\Big( \sum_{i=1}^{m} \omega ( r_{i-1}, r_{i} ) \Big)^{n}
\leqslant
\omega ( s, t ) ^{n},
$$
and hence we get the above inequality.
\end{Eg} 

\begin{Prop} 
\label{new-control} 
Let $\omega_{0}$ and $\omega$ be two control functions.
Let
$x_{0}:[0, +\infty ) \to \mathbb{R}$
be a continuous function
controlled by $\omega_{0}$
and with $x_{0} (0) = 0$.
Then
\begin{itemize}
\item[(i)]
$
\omega^{\prime} (s,t)
:=
\mathrm{e}^{\omega_{0}(s,t)}
(
	\omega_{0} (s,t) + \omega (s,t)
)
,
$
for
$0 \leqslant s \leqslant t$
defines a control function.

\end{itemize}
Let $n \in \mathbb{N}$
and
$
y_{1}, \cdots , y_{n} :
[ 0, +\infty ) \to \mathbb{C},
$
be continuous functions controlled by $\omega$.
Then
\begin{itemize}
\item[(ii)]
we have
\begin{equation*}
\begin{split}
\mathrm{e}^{ n x_{0}(t) }
\big\vert
\int_{
	0 \leqslant u_{1} < \cdots < u_{n} \leqslant t
}
\mathrm{d} y_{1} (u_{1})
\cdots
\mathrm{d} y_{n} (u_{n})
\big\vert
\leqslant
\frac{ \omega^{\prime} (0,t)^{n} }{ n! }.
\end{split}
\end{equation*}

\vspace{2mm}
\item[(iii)]
For each $0 \leqslant s \leqslant t \leqslant T$,
we have
\begin{equation*}
\begin{split}
&
\Big\vert
	\mathrm{e}^{ n x_{0} (t) }
	\int_{ 0 \leqslant u_{1} < \cdots < u_{n} \leqslant t }
	\mathrm{d} y_{1} (u_{1})
	\cdots
	\mathrm{d} y_{n} (u_{n})
	-
	\mathrm{e}^{ n x_{0} (s) }
	\int_{ 0 \leqslant u_{1} < \cdots < u_{n} \leqslant s }
	\mathrm{d} y_{1} (u_{1})
	\cdots
	\mathrm{d} y_{n} (u_{n})
\Big\vert \\
&\leqslant
\Big(
	\omega_{0} (s,t)
	+
	\omega (s,t)
\Big)
\Big(
	\omega^{\prime} (0,T)
	+
	\mathrm{e}^{ \omega_{0} (0,T) }
\Big)
\frac{ \omega^{\prime} (0,T)^{n-1} }{ (n-1)! } .
\end{split}
\end{equation*}
\end{itemize}
\end{Prop} 

\begin{proof} 
(i)
It is enough to show the super-additivity.
Put
$\omega^{\prime\prime} := \omega_{0} + \omega$.
Let $0 \leqslant s \leqslant u \leqslant t$ be arbitrary.
Then
\begin{equation*}
\begin{split}
\omega^{\prime} (s,u) + \omega^{\prime} (u,t)
&=
\mathrm{e}^{ \omega_{0} (s,u) }
\omega^{
\prime\prime
} (s,u)
+
\mathrm{e}^{ \omega_{0} (u,t) }
\omega^{
\prime\prime
} (u,t) \\
&=
\mathrm{e}^{ \omega_{0} (s,t) }
\big\{
	\mathrm{e}^{ \omega_{0} (s,u) - \omega_{0} (s,t) }
	\omega^{
	\prime\prime
	} (s,u)
	+
	\mathrm{e}^{ \omega_{0} (u,t) - \omega_{0} (s,t) }
	\omega^{
	\prime\prime
	} (u,t)
\big\} .
\end{split}
\end{equation*}
By the super-additivity and non-negativity of $\omega_{0}$, we have
$
\omega_{0} (s,u) - \omega_{0} (s,t) \leqslant 0
$
and
$
\omega_{0} (u,t) - \omega_{0} (s,t) \leqslant 0
$.
Therefore, by using the non-negativity and super-additivity for
$
\omega^{
\prime\prime
}
$,
we get
\begin{equation*}
\begin{split}
&
\omega^{\prime} (s,u) + \omega^{\prime} (u,t)
\leqslant
\mathrm{e}^{ \omega_{0} (s,t) }
\big\{
	\omega^{
	\prime\prime
	} (s,u)
	+
	\omega^{
	\prime\prime
	} (u,t)
\big\}
\leqslant
\mathrm{e}^{ \omega_{0} (s,t) }
\omega^{
\prime\prime
} (s,t)
=
\omega^{\prime} (s,t).
\end{split}
\end{equation*}

(ii)
Since
$
x_{0} = 0
$,
we have
$
\mathrm{e}^{n x_{0}(t)}
=
( \mathrm{e}^{x_{0}(t) - x_{0}(0)} )^{n}
\leqslant
\mathrm{e}^{ n \omega_{0} (0,t) }
$.
On the other hand, by
Example \ref{iterated:controlled} we have
that
\begin{equation*}
\begin{split}
\big\vert
\int_{
	0 \leqslant u_{1} < \cdots < u_{n} \leqslant t
}
\mathrm{d} y_{1} (u_{1})
\cdots
\mathrm{d} y_{n} (u_{n})
\big\vert
\leqslant
\frac{ \omega (0,t)^{n} }{ n! }.
\end{split}
\end{equation*}
Hence the assertion is immediate.

(iii)
Let $0 \leqslant s \leqslant t \leqslant T$ be arbitrary.
Then we have
\begin{equation*}
\begin{split}
&
\Big\vert
	\mathrm{e}^{ n x_{0} (t) }
	\int_{ 0 \leqslant u_{1} < \cdots < u_{n} \leqslant t }
	\mathrm{d} y_{1} (u_{1})
	\cdots
	\mathrm{d} y_{n} (u_{n})
	-
	\mathrm{e}^{ n x_{0} (s) }
	\int_{ 0 \leqslant u_{1} < \cdots < u_{n} \leqslant s }
	\mathrm{d} y_{1} (u_{1})
	\cdots
	\mathrm{d} y_{n} (u_{n})
\Big\vert \\
&\leqslant
\vert
	\mathrm{e}^{ n x_{0} (s) } - \mathrm{e}^{ n x_{0} (t) }
\vert
\times
\Big\vert
	\int_{ 0 \leqslant u_{1} < \cdots < u_{n} \leqslant t }
	\mathrm{d} y_{1} (u_{1})
	\cdots
	\mathrm{d} y_{n} (u_{n})
\Big\vert \\
&\hspace{10mm}+
\mathrm{e}^{ x_{0} (s) }
\int_{s}^{t}
\mathrm{e}^{ (n-1) x_{0} (s) }
\Big\vert
	\int_{ 0 \leqslant u_{1} < \cdots < u_{n-1} \leqslant u_{n} }
	\mathrm{d} y_{1} (u_{1})
	\cdots
	\mathrm{d} y_{n} (u_{n})
\Big\vert
\vert
	\mathrm{d} y_{n} ( u_{n} )
\vert .
\end{split}
\end{equation*}
Since it holds that
$
\vert
	\mathrm{e}^{ n x_{0} (s) }
	-
	\mathrm{e}^{ n x_{0} (t) }
\vert
\leqslant
n \omega_{0} (s,t)
\mathrm{e}^{ n \omega_{0} (0,t) },
$
and by using (ii), the above quantity is bounded by
\begin{equation*}
\begin{split}
&
n \omega_{0} (s,t)
\mathrm{e}^{ n \omega_{0} (0,t) }
\frac{ \omega (0,t)^{n} }{ n! }
+
\mathrm{e}^{ \omega_{0} (0,t) }
\omega (s,t)
\frac{ \omega^{\prime} (0,t)^{n-1} }{ (n-1)! } \\
&=
\Big(
	\omega_{0} (s,t)
	\omega^{\prime} (0,t)
	+
	\omega (s,t)
	\mathrm{e}^{ \omega_{0} (0,t) }
\Big)
\frac{ \omega^{\prime} (0,t)^{n-1} }{ (n-1)! } .
\end{split}
\end{equation*}
\end{proof} 

We shall remark here that
the control functions form a convex cone,
namely,
(a)
the sum of two control functions is a control function,
(b)
any control function multiplied by a positive real constant is
again a control function.
Therefore, the quantities in
Proposition~\ref{new-control}--(ii, iii)
are estimated by using a single control function.
Therefore, the following is immediate.

\begin{Cor} 
\label{Drive>>CLK} 
Let
$\omega_{0}$ and $\omega$
be two control functions.
Consider the  Loewner-Kufarev equation
{\rm (\ref{controlled-LK})}
and
suppose that
the following two conditions hold{\rm :}
\begin{itemize}
\item[(i)]
$x_{0}$ is controlled by $\omega_{0}${\rm ;}

\vspace{2mm}
\item[(ii)]
for every $n \in \mathbb{N}$,
there exist
continuous functions
$
y_{1}^{n}, y_{2}^{n}, \cdots , y_{n}^{n}:
[0,T] \to \mathbb{C}
$
controlled by $\omega$
such that
$$
x_{n}(t)
=
\int_{
	0
	\leqslant s_{1}
	< s_{2}
	< \cdots
	< s_{n}
	\leqslant t
}
\mathrm{d} y_{1}^{n} (s_{1})
\mathrm{d} y_{2}^{n} (s_{2})
\cdots
\mathrm{d} y_{n}^{n} (s_{n}),
\quad
0 \leqslant t \leqslant T.
$$
\end{itemize}
Then there exists a constant $c>0$
such that
{\rm (\ref{controlled-LK})}
is a Loewner-Kufarev equation controlled by
$
\omega^{\prime}
:=
c
( \omega_{0} + \omega )
\exp ( \omega_{0} )
$.
\end{Cor} 

The following is a consequence of
\cite[Theorem~2.8]{AmbFr}
and will be proved in
Section~\ref{app:bdd-sol}.

\begin{Cor} 
\label{boundedness(sol)} 
Let $\omega$ be a control function and
$\{  f_{t} \}_{0 \leqslant t \leqslant T}$
be a solution to the Loewner-Kufarev equation
controlled by $\omega$.
If
$\omega (0,T) < \frac{1}{4}$
then
$f_{t} ( \mathbb{D} )$
is bounded for any $t \in [0,T]$.
\end{Cor} 

\subsection{Some analytic aspects of Grunsky coefficients}

Let
$
S, S^{\prime} \subset \mathbb{Z}
$
be countably infinite subsets,
and
$A=(a_{i,j})_{i \in S, j \in S^{\prime}}$
be an
$S \times S^{\prime}$-matrix.
For each sequence
$x = (x_{j})_{j \in S^{\prime}}$
of complex numbers, we define a sequence
$
T_{A} x = ( (T_{A}x)_{i} )_{i \in S}
$
by
$
(T_{A}x)_{i} := \sum_{j \in S^{\prime}} a_{ij} x_{j}
$
if it converges for all $i \in S$.
We will still denote $T_{A}x$ by $Ax$  when it is defined.

Let $\ell_{2} (S)$ be the Hilbert space
consisting of all sequences
$a = (a_{i})_{i \in S}$
such that
$
\sum_{i \in S} \vert a_{i} \vert^{2}
< +\infty
$,
with the Hermitian inner product
$
\langle a, b \rangle_{2}
=
\sum_{i \in S} a_{i} \overline{b_{i}},
$
for
$a=(a_{i})_{i \in S}$,
$b=(b_{i})_{i \in S} \in \ell_{2} (S)$.
The associated norm will be denoted by
$\Vert \bullet \Vert_{2}$.

For each
$s \in \mathbb{R}$,
the space
$$
\ell_{2}^{
s
} (S)
:=
\Big\{
	a = (a_{n})_{n \in S} :
	\sum_{n \in S}
	(1 + n^{2} )^{
	s
	}
	\vert a_{n} \vert^{2} < +\infty
\Big\}
$$
is a Hilbert space under the Hermitian inner product
given by
\begin{equation*}
\langle a, b \rangle_{2,s}
:=
\sum_{n \in S}
\max\{ 1, \vert n \vert \}^{2s}
a_{n} \overline{b_{n}}
\end{equation*}
for
$a=(a_{i})_{i \in S}$,
$
b=(b_{i})_{i \in S}
\in \ell_{2}^{
s
} (S)
$.
The associated norm will be denoted by
$
\Vert \bullet \Vert_{2,
s
}
$.

Let us  recall a classical and  well-known result
from the theory of univalent functions.
For the definition and properties of Grunsky coefficients,
see
\cite[Chapter~2, Section~2]{TaTe06},
\cite[Section~2.2]{Te03}
or
\cite[Definition~A.1 and Proposition~A.2]{AmbFr}.

\begin{Thm}[(Grunsky's inequality {\cite[Theorem~3.2]{Po2}})] 
\label{Grunsky-ineq} 
Let $f : \mathbb{D} \to \mathbb{C}$ be a univalent functions
with $f(0) = 0$,
and let $(b_{m,n})_{m,n \leqslant -1}$
be the Grunsky coefficients associated to $f$.
Then for any
$m \in \mathbb{N}$
and
$
\lambda_{-m}, \lambda_{-m+1} , \cdots , \lambda_{-1} \in \mathbb{C}
$,
it holds that
\begin{equation*}
\begin{split}
\sum_{k \leqslant -1}
(-k)
\Big\vert
	\sum_{l=-m}^{-1}
	b_{k,l} \lambda_{l}
\Big\vert^{2}
\leqslant
\sum_{k=-m}^{-1}
\frac{ \vert \lambda_{k} \vert^{2} }{ (-k) } .
\end{split}
\end{equation*}

\end{Thm} 

This can be rephrased with our notation as follows:
Let
$
B := (
\sqrt{m(-n)}\,
b_{-m,n})_{m \in \mathbb{N}, n \in -\mathbb{N}}
$,
where $b_{m,n}$ for $m,n \leqslant -1$
are Grunsky coefficients associated to a univalent function
$f$
on $\mathbb{D}$ such that $f(0) = 0$.

\begin{Cor} 
\label{B-bounded} 
\begin{itemize}
\item[(i)]
$
B
:
\ell_{2} (-\mathbb{N})
\to
\ell_{2} (\mathbb{N})
$
and is a bounded linear operator with the operator norm satisfying
$\Vert B \Vert \leqslant 1$.

\vspace{2mm}
\item[(ii)]
$
B^{*}
:
\ell_{2} (\mathbb{N})
\to
\ell_{2} (-\mathbb{N})
$
and is a bounded linear operator with the operator norm satisfying
$\Vert B^{*} \Vert \leqslant 1$.

\vspace{2mm}
\item[(iii)]
The bounded linear operator
$
1 + B_{t} B_{t}^{*} :
\ell_{2} (\mathbb{N})
\to
\ell_{2} (\mathbb{N})
$
is injective and has a dense image.
\end{itemize}
\end{Cor} 

\begin{proof} 
(i)
For each
$
a = ( \cdots , a_{-3}, a_{-2}, a_{-1} )
\in \ell_{2} ( -\mathbb{N} )
$,
we have by Theorem~\ref{Grunsky-ineq},
\begin{equation*}
\begin{split}
\Vert B a \Vert_{2}^{2}
&= 
\sum_{n=1}^{\infty}
\Big(
	\sum_{k=1}^{\infty}
	\sqrt{nk}\,
	b_{-n,-k}
	a_{-k}
\Big)
\Big(
	\overline{
	\sum_{l=1}^{\infty}
	\sqrt{nl}\,
	b_{-n,-l} a_{-l}
	}
\Big) \\
&= 
\sum_{n=1}^{\infty}
n
\Big\vert
	\sum_{k=1}^{\infty}
	b_{-n,-k}
	( \sqrt{k}\, a_{-k} )
\Big\vert^{2}
\leqslant 
\sum_{n=1}^{\infty}
\frac{ \vert \sqrt{n}\, a_{-n} \vert^{2} }{ n }
=
\Vert a \Vert_{2}^{2} .
\end{split}
\end{equation*}

(ii)
Since the Grunsky matrix
$
(b_{m,n})_{m,n \leqslant -1}
$
is symmetric:
$
b_{m,n} = b_{n,m}
$
for all $m,n \leqslant -1$,
the assertion is proved similarly to (i).

(iii)
The injectivity is clear since the adjoint operator
of $B$ is $B^{*}$.
Then the second assertion is also clear
since $1+BB^{*}$ is self-adjoint.
\end{proof} 

\begin{Rm} 
The
semi-infinite matrix defined by
$
B_{1} := ( \sqrt{mn} b_{-m,-n} )_{m,n \in \mathbb{N}}
$
is called the {\it Grunsky operator},
and then the Grunsky inequality
(Theorem \ref{Grunsky-ineq})
shows that $B_{1}$ is a bounded operator on $\ell_{2} (\mathbb{N})$
with operator norm $\leqslant 1$.
This operator, together with three additional Grunsky operators,
are known to play a fundamental role in the study of the geometry of
the universal Teichm\"uller space.
For details, cf.   the papers by
Takhtajan-Teo~\cite{TaTe06}
or
Krushkal~ \cite{Kru07}.

\end{Rm} 

In the sequel, we fix a control function $\omega$,
and a solution
$\{ f_{t} \}_{0 \leqslant t \leqslant T}$
to a Loewner-Kufarev equation
controlled by $\omega$.
We denote by
$b_{m,n}(t)$ for $m,n \leqslant -1$
the Grunsky coefficients associated with $f_{t}$,
and
\begin{equation*}
\begin{split}
B_{t}
&:=
\big(
	\sqrt{m (-n)}\,
	b_{-m,n}(t)
\big)_{m \in \mathbb{N}, n \in -\mathbb{N}}, \\
B_{t}^{*}
&:=
\big(
	\sqrt{(-n) m}\,
	b_{ n, -m }(t)
\big)_{ n \in -\mathbb{N}, m \in \mathbb{N} }.
\end{split}
\end{equation*}

It
is clear that the linear operator
$
( 1 + B_{t} B_{t}^{*} )^{-1} :
\mathrm{Im}( 1 + B_{t} B_{t}^{*} )
\to
\ell_{2} (\mathbb{N})
$
is bounded.
Therefore by
Corollary~\ref{B-bounded}--(iii),
$( 1 + B_{t} B_{t}^{*} )^{-1}$
extends to $\ell_{2} (\mathbb{N})$
and the extension will be denoted by
$
A_{t} :
\ell_{2} (\mathbb{N})
\to
\ell_{2} (\mathbb{N})
$.
In particular, it is easy to see that 
$
\Vert A_{t} \Vert \leqslant 1,
$
holds for the operator norm.

We shall exhibit the indices which parametrise
our operators in order to help understanding the following:
\begin{equation*}
\begin{split}
B
=
\bordermatrix{
~ & \cdots & -3 & -2 & -1 \cr
1 & \cdots & B_{1,-3} & B_{1,-2} & B_{1,-1} \cr
2 & \cdots & B_{2,-3} & B_{2,-2} & B_{2,-1} \cr
3 & \cdots & B_{3,-3} & B_{3,-2} & B_{3,-1} \cr
\vdots & \rotatebox{70}{$\ddots$} & \vdots & \vdots & \vdots \cr},
\quad
B^{*}
=
\bordermatrix{
~ & 1 & 2 & 3 & \cdots \cr
\hspace{2mm} \vdots & \vdots & \vdots & \vdots & \rotatebox{70}{$\ddots$} \cr
-3 & B_{-3,1}^{*} & B_{-3,2}^{*} & B_{-3,3}^{*} & \cdots \cr
-2 & B_{-2,1}^{*} & B_{-2,2}^{*} & B_{-2,3}^{*} & \cdots \cr
-1 & B_{-1,1}^{*} & B_{-1,2}^{*} & B_{-1,3}^{*} & \cdots \cr},
\end{split}
\end{equation*}

\begin{equation*}
\begin{split}
BB^{*}
=
\bordermatrix{
~ & 1 & 2 & 3 & \cdots \cr
1 & * & * & * & \cdots \cr
2 & * & * & * & \cdots \cr
3 & * & * & * & \cdots \cr
\vdots & \vdots & \vdots & \vdots & \ddots \cr},
\quad
A = ( I + BB^{*} )^{-1}
=
\bordermatrix{
~ & 1 & 2 & 3 & \cdots \cr
1 & * & * & * & \cdots \cr
2 & * & * & * & \cdots \cr
3 & * & * & * & \cdots \cr
\vdots & \vdots & \vdots & \vdots & \ddots \cr}
\end{split}
\end{equation*}

and

\begin{equation*}
\begin{split}
&
B^{*}AB
=
\bordermatrix{
~ & \cdots & -3 & -2 & -1 \cr
\hspace{2mm}\vdots & \ddots & \vdots & \vdots & \vdots \cr
-3 & \cdots & (B^{*}AB)_{-3,-3} & (B^{*}AB)_{-3,-2} & (B^{*}AB)_{-3,-1} \cr
-2 & \cdots & (B^{*}AB)_{-2,-3} & (B^{*}AB)_{-2,-2} & (B^{*}AB)_{-2,-1} \cr
-1 & \cdots & (B^{*}AB)_{-1,-3} & (B^{*}AB)_{-1,-2} & (B^{*}AB)_{-1,-1} \cr} .
\end{split}
\end{equation*}

The following is a consequence from
\cite[Theorem~2.12]{AmbFr}
and will be proved in
Section~\ref{app:Grunsky-estimate}.

\begin{Cor} 
\label{Loew-Grunsky-estimate} 
Let $\omega$ be a control function,
and
$\{ f_{t} \}_{0 \leqslant t \leqslant T}$
be a solution to the Loewner-Kufarev equation
controlled by $\omega$.
Let
$
b_{-m,-n}(t)$, $n,m \in \mathbb{N}
$
be the Grunsky coefficients
associated to $f_{t}$,
for $0 \leqslant t \leqslant T$.
Then for any
$0 \leqslant s \leqslant t \leqslant T$
and
$n,m \in \mathbb{N}$
with $n+m \geqslant 3$,
we have
\begin{itemize}
\item[(i)]
$
\vert b_{-1,-1} (t) \vert
\leqslant
\frac{ \omega (0,t)^{2} }{ 2 }
$
and
$
\vert b_{-1,-1} (t) - b_{-1,-1} (s) \vert
\leqslant
\omega (s,t) \omega (0,T)
$.

\vspace{2mm}
\item[(ii)]
$\displaystyle
\vert
	b_{-m,-n} (t)
\vert
\leqslant
\frac{
	( 8 \omega (0,t) )^{m+n}
}{
	16 (m+n) (m+n-1) (m+n-2)
}
$.

\vspace{2mm}
\item[(iii)]
$\displaystyle
\vert
	b_{-m,-n} (t)
	-
	b_{-m,-n} (s)
\vert
\leqslant
\frac{
	\omega (s,t) ( 8 \omega (0,T) )^{m+n-1}
}{
	16 (m+n-1) (m+n-2)
}
$.
\end{itemize}
\end{Cor} 

Along the Loewner-Kufarev equation
controlled by $\omega$,
we obtain the following

\begin{Cor} 
\label{modulus-B-estimate} 
If $\omega (0,T) < \frac{1}{8}$, then
for $0 \leqslant s < t \leqslant T$,
\begin{itemize}
\item[(i)]
$\displaystyle
\Vert
	B_{t}^{*} - B_{s}^{*}
\Vert
=
\Vert
	B_{t} - B_{s}
\Vert
\leqslant
c
\hspace{0.5mm}
\omega (s,t)
$,

\vspace{2mm}
\item[(ii)]
$\displaystyle
\Vert
	A_{t} - A_{s}
\Vert
\leqslant
2c
\hspace{0.5mm}
\omega (s,t)
$,
\end{itemize}
where
$
c
:=
\frac{ 8 \omega (0,T) }{ 1 - ( 8 \omega (0,T) )^{2} }
> 0
$.
\end{Cor} 
\begin{proof} 
(i)
By Corollary \ref{Loew-Grunsky-estimate}--(ii),
we have
\begin{equation*}
\begin{split}
&
\Vert
	B_{t} - B_{s}
\Vert^{2}
\leqslant 
\sum_{n=1}^{\infty} \sum_{m=1}^{\infty}
\vert
	( B_{t} - B_{s} )_{n,-m}
\vert^{2} \\
&= 
\sum_{n=1}^{\infty} \sum_{m=1}^{\infty}
\vert
	\sqrt{nm}
	( b_{-n.-m} (t) - b_{-n,-m} (s) )
\vert^{2} \\
&\leqslant 
\sum_{n=1}^{\infty} \sum_{m=1}^{\infty}
\Big\vert
	\frac{
		\sqrt{nm}
		\omega (s,t)
		( 8 \omega (0,T) )^{n+m-1}
	}{
		16 (n+m-1) (n+m-2)
	}
\Big\vert^{2} \\
&\leqslant 
\Big(
	\frac{ \omega (s,t) }{ 8 \omega (0,T) }
\Big)^{2}
\Big(
\sum_{n=1}^{\infty}
( 8 \omega (0,T) )^{2n}
\Big)^{2}
= 
\Big(
	\omega (s,t)
	\frac{ 8 \omega (0,T) }{ 1 - ( 8 \omega (0,T) )^{2} }
\Big)^{2} .
\end{split}
\end{equation*}

(ii)
Since
\begin{equation*}
\begin{split}
&
A_{t} - A_{s}
=
( 1 + B_{t} B_{t}^{*} )^{-1}
-
( 1 + B_{s} B_{s}^{*} )^{-1} \\
&=
( 1 + B_{t} B_{t}^{*} )^{-1}
( B_{s} B_{s}^{*} - B_{t} B_{t}^{*} )
( 1 + B_{s} B_{s}^{*} )^{-1} \\
&=
( 1 + B_{t} B_{t}^{*} )^{-1}
( B_{s} - B_{t} )
B_{t}^{*}
( 1 + B_{s} B_{s}^{*} )^{-1} \\
&\hspace{20mm}-
( 1 + B_{t} B_{t}^{*} )^{-1}
B_{s}
( B_{t}^{*} - B_{s}^{*} )
( 1 + B_{s} B_{s}^{*} )^{-1},
\end{split}
\end{equation*}
we have
$
\Vert A_{t} - A_{s} \Vert
\leqslant
\Vert
	B_{s} - B_{t}
\Vert
+
\Vert
	B_{s}^{*} - B_{t}^{*}
\Vert
=
2
\Vert
	B_{t} - B_{s}
\Vert
$.
\end{proof} 

Finally, define
$
\Lambda
=
( \Lambda_{m,n} )_{
m \in
\mathbb{Z},
n \in
\mathbb{Z}
}
$
by
$
\Lambda_{m,n}
:=
\sqrt{m}\,
\delta_{m,-n}
+
\delta_{m,0}
\delta_{0,n}
$
for
$m \in \mathbb{N}$
and
$n \in -\mathbb{N}$,
that is,
\begin{equation}
\label{mat-Lambda} 
\begin{split}
\Lambda
=
\left(
\begin{array}{c:c:c}
\begin{array}{cc}
\begin{matrix}
\ddots & \\
       & \sqrt{3}
\end{matrix}
&  \\
&
\begin{matrix}
\sqrt{2} & \\
         & \sqrt{1}
\end{matrix}
\end{array}
&  &  \\ \hdashline
 & 1 &  \\ \hdashline
& 
&
\begin{array}{cc}
\begin{matrix}
\sqrt{1} & \\
         & \sqrt{2}
\end{matrix}
& \\
&
\begin{matrix}
\sqrt{3} & \\
         & \ddots
\end{matrix}
\end{array}
\end{array}\right) .
\end{split}
\end{equation}

It is clear that
$
\Lambda :
\ell_{2}^{
1/2
}
(
\mathbb{Z}
)
\to
\ell_{2}
(
\mathbb{Z}
)
$
and is a continuous linear isomorphism.

\section{Proof of Theorem~\ref{modulus-LK}}
\label{Sec_LK/Gr} 

Let $\omega$ be a control function such that $\omega (0,T) < \frac{1}{8}$,
and
let
$\{ f_{t} \}_{0 \leqslant t \leqslant T}$
be a univalent solution to the Loewner-Kufarev equation
controlled by $\omega$.
Suppose further that $f_{t}$ extends to holomorphic functions
on open neighbourhoods of $\overline{\mathbb{D}}$
for all $t \in [0,T]$.

We then note that for each $t \in [0,T]$,
it holds that
$Q_{n} ( t, f_{t} (1/z) )\vert_{S^{1}} \in H^{1/2}$,
where
$Q_{n} ( t, w )$ is the $n$-th Faber polynomial associated to $f_{t}$.
Therefore we have
$$
\mathrm{span}
\big(
	\{ 1 \}
	\cup
	\{ Q_{n} ( t, f_{t} (1/z) )\vert_{S^{1}} \}_{n \geqslant 1}
\big)
\subset
H^{1/2}
\subset
H.
$$
In particular, we have
\begin{equation*}
\begin{split}
W_{f_{t}}^{1/2}
&:=
\overline{
\mathrm{span}
\big(
	\{ 1 \}
	\cup
	\{ Q_{n} ( t, f_{t} (1/z) )\vert_{S^{1}} \}_{n \geqslant 1}
\big)
}^{H^{1/2}} \\
&\subset
\overline{
\mathrm{span}
\big(
	\{ 1 \}
	\cup
	\{ Q_{n} ( t, f_{t} (1/z) )\vert_{S^{1}} \}_{n \geqslant 1}
\big)
}^{H}
= W_{f_{t}}.
\end{split}
\end{equation*}

We fix an inner product on $H^{1/2}$ by requiring for
$
h = \sum_{n \in \mathbb{Z}} h_{n} z^{n},
g = \sum_{n \in \mathbb{Z}} g_{n} z^{n}
\in H^{1/2}
$,
that
$
\langle h, g \rangle_{H^{1/2}}
:=
h_{0} \overline{g}_{0}
+
\sum_{n=1}^{\infty}
n
(
	h_{-n} \overline{g}_{-n}
	+
	h_{n} \overline{g}_{n}
)
$.
Then
$
\{ \frac{ z^{-n} }{ \sqrt{n} } \}_{n \in \mathbb{N}}
\cup \{ 1 \}
\cup \{ \frac{ z^{n} }{ \sqrt{n} } \}_{n \in \mathbb{N}}
$
forms a complete orthonormal system of $H^{1/2}$.
By this, the infinite matrix $\Lambda$ defined in (\ref{mat-Lambda})
determines a bounded linear isomorphism
$
H^{1/2} \to H
$
through the identification $H^{1/2} \cong \ell_{2}^{1/2}(\mathbb{Z})$.

Recall that for each univalent function
$f : \mathbb{D} \to \mathbb{C}$
with $f(0) = 0$
and an analytic continuation across $S^{1}$,
the orthogonal projection
$H^{1/2} \to W_{f}^{1/2}$
is denoted by $\EuScript{P}_{f}$.
In order to prove Theorem \ref{modulus-LK},
we need to calculate the projection operator
$\EuScript{P}_{f}$.
For this, we shall consider first the following change of basis.


Let
$\mathbf{w}_{n} (z) := Q_{n} \circ f (z^{-1})$,
for
$z \in S^{1}$
and
$n \in \mathbb{N}$.
Then we have
\begin{equation}
\label{rep_mat} 
\begin{split}
&
(
	\cdots ,
	\frac{ \mathbf{w}_{3} }{ \sqrt{3} },
	\frac{ \mathbf{w}_{2} }{ \sqrt{2} },
	\frac{ \mathbf{w}_{1} }{ \sqrt{1} }
) \\
&=
\Big(\begin{array}{cccc:c:cccc}
	\cdots ,
	&
	\frac{ z^{3} }{ \sqrt{3} },
	&
	\frac{ z^{2} }{ \sqrt{2} },
	&
	\frac{ z^{1} }{ \sqrt{1} },
	& 1,
	&
	\frac{ z^{-1} }{ \sqrt{1} },
	&
	\frac{ z^{-2} }{ \sqrt{2} },
	&
	\frac{ z^{-3} }{ \sqrt{3} },
	& \cdots
\end{array}\Big) \\
&\hspace{10mm}\times
\left(
\begin{array}{ccccc}
\rotatebox{10}{$\ddots$} & \vdots & \vdots & \vdots \\
\cdots & 1 & 0 & 0 \\
\cdots & 0 & 1 & 0 \\
\cdots & 0 & 0 & 1 \\ \hdashline
\cdots & 0 & 0 & 0 \\ \hdashline
\cdots &
	\sqrt{3\cdot 1}\,
b_{-3,-1}
&
	\sqrt{2\cdot 1}\,
b_{-2,-1}
&
	\sqrt{1\cdot 1}\,
b_{-1,-1} \\
\cdots
&
	\sqrt{3\cdot 2}\,
b_{-3,-2}
&
	\sqrt{2\cdot 2}\,
b_{-2,-2}
&
	\sqrt{1\cdot 2}\,
b_{-1,-2} \\
\cdots
&
	\sqrt{3\cdot 3}\,
b_{-3,-3}
&
	\sqrt{2\cdot 3}\,
b_{-2,-3}
&
	\sqrt{1\cdot 3}\,
b_{-1,-3} \\
\rotatebox{70}{$\ddots$} & \vdots & \vdots & \vdots
\end{array}\right)
\end{split}
\end{equation}

By putting
\begin{equation*}
\begin{split}
\widetilde{\mathbf{z}}_{+}
&:=
\Big(
	\cdots , \frac{z^{3}}{\sqrt{3}}, \frac{z^{2}}{\sqrt{2}}, \frac{z}{\sqrt{1}}
\Big),
\quad
\widetilde{\mathbf{z}}_{-}
:=
\Big(
		\frac{z^{-1}}{\sqrt{1}},
		\frac{z^{-2}}{\sqrt{2}},
		\frac{ z^{-3} }{ \sqrt{3} },
	\cdots
\Big),
\end{split}
\end{equation*}
and
\begin{equation*}
\begin{split}
B:=
\left(
\begin{array}{ccccc}
\cdots
&
	\sqrt{3\cdot 1}\,
b_{-3,-1}
&
	\sqrt{2\cdot 1}\,
b_{-2,-1}
&
	\sqrt{1\cdot 1}\,
b_{-1,-1} \\
\cdots
&
	\sqrt{3\cdot 2}\,
b_{-3,-2}
&
	\sqrt{2\cdot 2}\,
b_{-2,-2}
&
	\sqrt{1\cdot 2}\,
b_{-1,-2} \\
\cdots
&
	\sqrt{3\cdot 3}\,
b_{-3,-3}
&
	\sqrt{2\cdot 3}\,
b_{-2,-3}
&
	\sqrt{1\cdot 3}\,
b_{-1,-3} \\
\rotatebox{70}{$\ddots$} & \vdots & \vdots & \vdots
\end{array}\right) ,
\end{split}
\end{equation*}
(so we have put
$
B_{t} = ( (B_{t})_{n,m} )_{n \geqslant 1, m \leqslant -1}
$
where
$(B_{t})_{n,m} :=
	\sqrt{n(-m)}\,
b_{m,-n} (t)$
for
$n \geqslant 1$
and
$m \leqslant -1$),
the equation (\ref{rep_mat}) is written in a simpler form as:
\begin{equation*}
\begin{split}
\Big(
	\cdots ,
	\frac{ \mathbf{w}_{3} }{ \sqrt{3} },
	\frac{ \mathbf{w}_{2} }{ \sqrt{2} },
	\frac{ \mathbf{w}_{1} }{ \sqrt{1} }
\Big)
=
(\begin{array}{c:c:c}
\widetilde{\mathbf{z}}_{+}
& 1 & \widetilde{\mathbf{z}}_{-}
\end{array})
\left(
\begin{array}{c}
I \\ \hdashline
\mathbf{0} \\ \hdashline
B
\end{array}\right) ,
\end{split}
\end{equation*}
where
\begin{equation*}
\begin{split}
\left(\begin{array}{c}
I \\ \hdashline
\mathbf{0} \\ \hdashline
B
\end{array}\right)
=
\left(
\begin{array}{ccccc}
\rotatebox{10}{$\ddots$} & \vdots & \vdots & \vdots \\
\cdots & 1 & 0 & 0 \\
\cdots & 0 & 1 & 0 \\
\cdots & 0 & 0 & 1 \\ \hdashline
\cdots & 0 & 0 & 0 \\ \hdashline
\cdots &
	\sqrt{3\cdot 1}\,
b_{-3,-1}
&
	\sqrt{2\cdot 1}\,
b_{-2,-1}
&
	\sqrt{1\cdot 1}\,
b_{-1,-1} \\
\cdots
&
	\sqrt{3\cdot 2}\,
b_{-3,-2}
&
	\sqrt{2\cdot 2}\,
b_{-2,-2}
&
	\sqrt{1\cdot 2}\,
b_{-1,-2} \\
\cdots
&
	\sqrt{3\cdot 3}\,
b_{-3,-3}
&
	\sqrt{2\cdot 3}\,
b_{-2,-3}
&
	\sqrt{1\cdot 3}\,
b_{-1,-3} \\
\rotatebox{70}{$\ddots$} & \vdots & \vdots & \vdots
\end{array}\right) .
\end{split}
\end{equation*}

Consider the change of basis
\begin{equation*}
\begin{split}
(\begin{array}{c:c:c}
\widetilde{\mathbf{w}}
& 1
&
\widetilde{\mathbf{v}}
\end{array})
:=:
(
	\cdots ,
	\frac{ \mathbf{w}_{3} }{ \sqrt{3} },
	\frac{ \mathbf{w}_{2} }{ \sqrt{2} },
	\frac{ \mathbf{w}_{1} }{ \sqrt{1} },
	1,
	\frac{ \mathbf{v}_{1} }{ \sqrt{1} },
	\frac{ \mathbf{v}_{2} }{ \sqrt{2} },
	\frac{ \mathbf{v}_{3} }{ \sqrt{3} },
	\cdots
)
:=
(\begin{array}{c:c:c}
	\widetilde{\mathbf{z}}_{+}
& 1 & \widetilde{\mathbf{z}}_{-}
\end{array})
\left(
\begin{array}{c:c:c}
I & \mathbf{0} & -B^{*} \\ \hdashline
\mathbf{0} & 1 & \mathbf{0} \\ \hdashline
B & \mathbf{0} & I
\end{array}\right)
\end{split}
\end{equation*}
where we note that the matrix on the right-hand side
is non-degenerate, with inverse
\begin{equation*}
\begin{split}
\left(
\begin{array}{c:c:c}
I & \mathbf{0} & -B^{*} \\ \hdashline
\mathbf{0} & 1 & \mathbf{0} \\ \hdashline
B & \mathbf{0} & I
\end{array}\right)^{-1}
=
\left(
\begin{array}{c:c:c}
I - B^{*} ( I + BB^{*} )^{-1} B & \mathbf{0} & B^{*} ( I + BB^{*} )^{-1} \\ \hdashline
\mathbf{0} & 1 & \mathbf{0} \\ \hdashline
- ( I + BB^{*} )^{-1} B & \mathbf{0} & ( I + BB^{*} )^{-1}
\end{array}\right) .
\end{split}
\end{equation*}

We note that the identity
\begin{equation*}
\begin{split}
\left(\begin{array}{c:c:c}
I & \mathbf{0} & -B^{*} \\ \hdashline
\mathbf{0} & 1 & \mathbf{0} \\ \hdashline
B & \mathbf{0} & I
\end{array}\right)
\left(\begin{array}{c:c:c}
I & \mathbf{0} & -B^{*} \\ \hdashline
\mathbf{0} & 1 & \mathbf{0} \\ \hdashline
B & \mathbf{0} & I
\end{array}\right)^{*}
=
\left(\begin{array}{c:c:c}
I + B^{*}B & \mathbf{0} & O \\ \hdashline
\mathbf{0} & 1 & \mathbf{0} \\ \hdashline
O & \mathbf{0} & I + B^{*}B
\end{array}\right)
\end{split}
\end{equation*}
and the fact that
$
(
	\widetilde{\mathbf{z}}_{+},
	1,
	\widetilde{\mathbf{z}}_{-}
)
$
is a complete orthonormal system in $H^{1/2}$
implies that
\begin{equation*}
\begin{split}
&
H^{1/2} = 
\overline{
\mathrm{span} \{ 1, \mathbf{w}_{1}, \mathbf{w}_{2}, \mathbf{w}_{3}, \cdots \}
}^{H^{1/2}}
\oplus
\overline{
\mathrm{span} \{ \mathbf{v}_{1}, \mathbf{v}_{2}, \mathbf{v}_{3}, \cdots \}
}^{H^{1/2}}
\end{split}
\end{equation*}
is an orthogonal decomposition of
$H^{1/2}$.

Let $A := ( I + BB^{*} )^{-1}$.
Then
\begin{equation*}
\begin{split}
(
	\widetilde{\mathbf{z}}_{+},
1, \widetilde{\mathbf{z}}_{-} )
&=
(
	\widetilde{\mathbf{w}},
1,
	\widetilde{\mathbf{v}}
)
\left(\begin{array}{c:c:c}
I & \mathbf{0} & -B^{*} \\ \hdashline
\mathbf{0} & 1 & \mathbf{0} \\ \hdashline
B & \mathbf{0} & I
\end{array}\right)^{-1} \\
&=
(
	\widetilde{\mathbf{w}},
1,
	\widetilde{\mathbf{v}}
)
\left(\begin{array}{c:c:c}
I - B^{*} A B & \mathbf{0} & B^{*} A \\ \hdashline
\mathbf{0} & 1 & \mathbf{0} \\ \hdashline
- AB & \mathbf{0} & A
\end{array}\right) \\
&=
(
		\widetilde{\mathbf{w}}
	( I - B^{*} A B )
	-
		\widetilde{\mathbf{v}}
	AB,
	1,
		\widetilde{\mathbf{w}}
	B^{*} A
	+
		\widetilde{\mathbf{v}}
	A
)
\end{split}
\end{equation*}
so that
\begin{equation*}
\begin{split}
(
	\EuScript{P}_{f}
	(
		\widetilde{\mathbf{z}}_{+}
	),
	1,
	\EuScript{P}_{f} ( \widetilde{\mathbf{z}}_{-} )
)
&:=
(
	\cdots ,
	\EuScript{P}_{f}
		\Big( \frac{ z^{2} }{ \sqrt{2} } \Big),
	\EuScript{P}_{f} (z),
	1,
	\EuScript{P}_{f} ( z^{-1} ),
	\EuScript{P}_{f}
		\Big( \frac{ z^{-2} }{ \sqrt{2} } \Big),
	\cdots
) \\
&=
(
		\widetilde{\mathbf{w}}
	( I - B^{*} A B ),
	1,
		\widetilde{\mathbf{w}}
	B^{*} A
) \\
&=
(
	(
		\widetilde{\mathbf{z}}_{+}
	+ \widetilde{\mathbf{z}}_{-} B ) ( I - B^{*} A B ),
	1,
	(
		\widetilde{\mathbf{z}}_{+}
	+ \widetilde{\mathbf{z}}_{-} B ) B^{*} A
) \\
&=
(
		\widetilde{\mathbf{z}}_{+},
	1,
	\widetilde{\mathbf{z}}_{-}
)
\left(\begin{array}{c:c:c}
I - B^{*} A B & \mathbf{0} & B^{*} A \\ \hdashline
\mathbf{0} & 1 & \mathbf{0} \\ \hdashline
B ( I - B^{*} A B ) & \mathbf{0} & B B^{*} A
\end{array}\right)
\end{split}
\end{equation*}

From these, we find that
for $n \geqslant 1$,
\begin{equation*}
\begin{split}
&
\EuScript{P}_{f}
	\Big( \frac{ z^{n} }{ \sqrt{n} } \Big)
=
\sum_{k=1}^{\infty}
\Big\{
	\frac{ z^{k} }{ \sqrt{k} }
( I - B^{*} A B )_{-k,-n}
+
\frac{
	z^{-k}
}{
	\sqrt{k}
}
( B ( I - B^{*} A B ) )_{k,-n}
\Big\} , \\
& 
\EuScript{P}_{f}
	\Big( \frac{ z^{-n} }{ \sqrt{n} } \Big)
=
\sum_{k=1}^{\infty}
\Big\{
	\frac{ z^{k} }{ \sqrt{k} }
( B^{*} A )_{-k,n}
+
\frac{
	z^{-k}
}{
		\sqrt{k}
}
( B B^{*} A )_{k,n}
\Big\} ,
\end{split}
\end{equation*}
from which, the following is immediate:

\begin{Prop} 
Let
$
h = \sum_{k \in \mathbb{Z}} h_{k} z^{k}
\in
H^{1/2}
$.
Then
\begin{equation*}
\begin{split}
\EuScript{P}_{f} (h)
&=
\sum_{n=1}^{\infty}
\Big\{
	\sum_{k=1}^{\infty}
	\big[
	( I - B^{*} A B )_{-n,-k}
	\sqrt{k}\,
	h_{k}
	+
	( B^{*} A )_{-n,k}
	\sqrt{k}\,
	h_{-k}
	\big]
\Big\}
	\frac{ z^{n} }{ \sqrt{n} } \\
&+
h_{0}
+
\sum_{n=1}^{\infty}
\Big\{
	\sum_{k=1}^{\infty}
	\big[
	( B ( I - B^{*} A B ) )_{n,-k}
	\sqrt{k}\,
	h_{k}
	+
	( B B^{*} A )_{n,k}
		\sqrt{k}\,
	h_{-k}
	\big]
\Big\}
	\frac{ z^{-n} }{ \sqrt{n} }.
\end{split}
\end{equation*}
\end{Prop} 

Denote by $B_{t}$ the associated matrix of Grunsky coefficients
of $f_{t}$.

\begin{Prop} 
\label{diff-norm} 
Let
$
h = \sum_{k \in \mathbb{Z}} h_{k} z^{k}
\in
	H^{1/2}
$.
Then
\begin{equation*}
\begin{split}
&
\EuScript{P}_{f_{t}} (h) - \EuScript{P}_{f_{s}} (h) \\
&=
\sum_{n=1}^{\infty}
\Big\{
	\sum_{k=1}^{\infty}
	( B_{s}^{*} A_{s} B_{s} - B_{t}^{*} A_{t} B_{t} )_{-n,-k}
		\sqrt{k}\,
	h_{k}
	+
	\sum_{k=1}^{\infty}
	( B_{t}^{*} A_{t} - B_{s}^{*} A_{s} )_{-n,k}
	\sqrt{k}\,
	h_{-k}
\Big\}
	\frac{ z^{n} }{ \sqrt{n} } \\
&+
\sum_{n=1}^{\infty}
\Big\{
	\sum_{k=1}^{\infty}
	\big(
	B_{t} ( I - B_{t}^{*} A_{t} B_{t} )
	-
	B_{s} ( I - B_{s}^{*} A_{s} B_{s} )
	\big)_{n,-k}
	\sqrt{k}\,
	h_{k} \\
	&\hspace{50mm}+
	\sum_{k=1}^{\infty}
	( B_{t} B_{t}^{*} A_{t} - B_{s} B_{s}^{*} A_{s} )_{n,k}
	\sqrt{k}\,
	h_{-k}
\Big\}
	\frac{ z^{-n} }{ \sqrt{n} } ,
\end{split}
\end{equation*}
so that
\begin{equation*}
\begin{split}
&
\Vert
	\EuScript{P}_{f_{t}} (h)
	-
	\EuScript{P}_{f_{s}} (h)
\Vert_{
		H^{1/2}
}^{2} \\
&=
\sum_{n=1}^{\infty}
\Big\vert
	\sum_{k=1}^{\infty}
	( B_{s}^{*} A_{s} B_{s} - B_{t}^{*} A_{t} B_{t} )_{-n,-k}
		\sqrt{k}\,
	h_{k}
	+
	\sum_{k=1}^{\infty}
	( B_{t}^{*} A_{t} - B_{s}^{*} A_{s} )_{-n,k}
		\sqrt{k}\,
	h_{-k}
\Big\vert^{2} \\
&+
\sum_{n=1}^{\infty}
\Big\vert
	\sum_{k=1}^{\infty}
	\big(
	B_{t} ( I - B_{t}^{*} A_{t} B_{t} )
	-
	B_{s} ( I - B_{s}^{*} A_{s} B_{s} )
	\big)_{n,-k}
		\sqrt{k}\,
	h_{k} \\
	&\hspace{50mm}+
	\sum_{k=1}^{\infty}
	( B_{t} B_{t}^{*} A_{t} - B_{s} B_{s}^{*} A_{s} )_{n,k}
		\sqrt{k}\,
	h_{-k}
\Big\vert^{2} .
\end{split}
\end{equation*}
\end{Prop} 

We are now in a position to prove
Theorem \ref{modulus-LK}.

\begin{proof}[Proof of Theorem~\ref{modulus-LK}]
By Proposition \ref{diff-norm},
we have
\begin{equation*}
\begin{split}
&
\Vert
	\EuScript{P}_{f_{t}} (h)
	-
	\EuScript{P}_{f_{s}} (h)
\Vert_{
H^{1/2}
}^{2}
\leqslant 
2 ( I + I\!\!I )
+
3
( I\!\!I\!\!I + I\!V + V ) ,
\end{split}
\end{equation*}
where
\begin{equation*}
\begin{split}
I
&:=
\sum_{n=1}^{\infty}
\Big\vert
	\sum_{k=1}^{\infty}
	( B_{s}^{*} A_{s} B_{s} - B_{t}^{*} A_{t} B_{t} )_{-n,-k}
		\sqrt{k}\,
	h_{k}
\Big\vert^{2}, \\
I\!\!I
&:=
\sum_{n=1}^{\infty}
\Big\vert
	\sum_{k=1}^{\infty}
	( B_{t}^{*} A_{t} - B_{s}^{*} A_{s} )_{-n,k}
		\sqrt{k}\,
	h_{-k}
\Big\vert^{2}, \\
I\!\!I\!\!I
&:=
\sum_{n=1}^{\infty}
\Big\vert
	\sum_{k=1}^{\infty}
	( B_{t} - B_{s} )_{n,-k}
	\sqrt{k}\,
	h_{k}
\Big\vert^{2}, \\
I\!V
&:=
\sum_{n=1}^{\infty}
\Big\vert
	\sum_{k=1}^{\infty}
	(
	B_{t}B_{t}^{*} A_{t} B_{t}
	-
	B_{s}B_{s}^{*} A_{s} B_{s}
	)_{n,-k}
	\sqrt{k}\,
	h_{k}
\Big\vert^{2}, \\
V
&:=
\sum_{n=1}^{\infty}
\Big\vert
	\sum_{k=1}^{\infty}
	( B_{t} B_{t}^{*} A_{t} - B_{s} B_{s}^{*} A_{s} )_{n,k}
	\sqrt{k}\,
	h_{-k}
\Big\vert^{2} .
\end{split}
\end{equation*}

\noindent
\underline{{\it Estimate of $I$.}}

\begin{equation*}
\begin{split}
I
&=
\sum_{n=1}^{\infty}
\Big\vert
	\sum_{k=1}^{\infty}
	( B_{s}^{*} A_{s} B_{s} - B_{t}^{*} A_{t} B_{t} )_{-n,-k}
	\sqrt{k}\,
	h_{k}
\Big\vert^{2}
\end{split}
\end{equation*}

We shall note that
$
A_{s} - A_{t}
=
A_{t} [ ( I + B_{t}B_{t}^{*} ) - ( I + B_{s}B_{s}^{*} ) ] A_{s}
=
A_{t} ( B_{t}B_{t}^{*} - B_{s}B_{s}^{*} ) A_{s}
$
and hence we obtain the following identity:
\begin{equation*}
\begin{split}
&
B_{s}^{*} A_{s} B_{s} - B_{t}^{*} A_{t} B_{t} \\
&=
( B_{s}^{*} - B_{t}^{*} ) A_{s} B_{s}
+ B_{t}^{*} ( A_{s} - A_{t} ) B_{s}
+ B_{t}^{*} A_{t} ( B_{s} - B_{t} ) .
\end{split}
\end{equation*}
According to this decomposition,
$I$ can be estimated as
\begin{equation*}
I
\leqslant
3 ( I_{1} + I_{2} + I_{3} ),
\end{equation*}
where
\begin{equation*}
\begin{split}
I_{1}
&:=
\sum_{n=1}^{\infty}
\Big\vert
	\sum_{k=1}^{\infty}
	( ( B_{s}^{*} - B_{t}^{*} ) A_{s} B_{s} )_{-n,-k}
	\sqrt{k}\,
	h_{k}
\Big\vert^{2}, \\
I_{2}
&:=
\sum_{n=1}^{\infty}
\Big\vert
	\sum_{k=1}^{\infty}
	( B_{t}^{*} ( A_{s} - A_{t} ) B_{s} )_{-n,-k}
	\sqrt{k}\,
	h_{k}
\Big\vert^{2}, \\
I_{3}
&:=
\sum_{n=1}^{\infty}
\Big\vert
	\sum_{k=1}^{\infty}
	( B_{t}^{*} A_{t} ( B_{s} - B_{t} ) )_{-n,-k}
	\sqrt{k}\,
	h_{k}
\Big\vert^{2} .
\end{split}
\end{equation*}
Each of which is estimated as follows:
By
Corollaries~\ref{B-bounded} and \ref{modulus-B-estimate},
we have
\begin{equation*}
\begin{split}
I_{1}
&\leqslant 
\Vert B_{s} - B_{t} \Vert^{2}
\Vert A_{s} B_{s} \Vert^{2}
\Vert
\Lambda
h
\Vert_{H}^{2}
\leqslant 
c_{11}
\omega (s,t)^{2}
\Vert h \Vert_{
H^{1/2}
}^{2}.
\end{split}
\end{equation*}
for some constant $c_{11} >0$.
Similarly, we have
\begin{equation*}
\begin{split}
I_{2}
&\leqslant
\Vert B_{t}^{*} \Vert^{2}
\Vert A_{s} - A_{t} \Vert^{2}
\Vert B_{t} \Vert^{2}
\Vert
\Lambda
h
\Vert_{H}^{2}
\leqslant
c_{12} \omega (s,t)^{2}
\Vert h \Vert_{
H^{1/2}
}^{2}, \\
I_{3}
&\leqslant
\Vert B_{t} A_{t} \Vert^{2}
\Vert B_{s} - B_{t} \Vert^{2}
\Vert
\Lambda
h
\Vert_{H}^{2}
\leqslant
c_{13} \omega (s,t)^{2}
\Vert h \Vert_{
H^{1/2}
}^{2}
\end{split}
\end{equation*}
for some constants $c_{12}, c_{13} > 0$.
Combining these together, we obtain
\begin{equation*}
I_{1}
\leqslant
c_{1} \omega (s,t)^{2}
\Vert h \Vert_{
H^{1/2}
}^{2}
\end{equation*}
for some constant $c_{1} > 0$.

\noindent
\underline{{\it Estimate of $I\!\!I$.}}

\begin{equation*}
\begin{split}
I\!\!I
&=
\sum_{n=1}^{\infty}
\Big\vert
	\sum_{k=1}^{\infty}
	( B_{t}^{*} A_{t} - B_{s}^{*} A_{s} )_{-n,k}
	\sqrt{k}\,
	h_{-k}
\Big\vert^{2}
\end{split}
\end{equation*}

According to the identity
\begin{equation*}
\begin{split}
B_{t}^{*} A_{t} - B_{s}^{*} A_{s}
=
( B_{t}^{*} - B_{s}^{*} ) A_{t}
+
B_{s}^{*} ( A_{t} - A_{s} ) ,
\end{split}
\end{equation*}
we estimate $I\!\!I$ as
\begin{equation*}
\begin{split}
I\!\!I
\leqslant
2( I\!\!I_{1} + I\!\!I_{2} ),
\end{split}
\end{equation*}
where
\begin{equation*}
\begin{split}
I\!\!I_{1}
&=
\sum_{n=1}^{\infty}
\Big\vert
	\sum_{k=1}^{\infty}
	( ( B_{t}^{*} - B_{s}^{*} )A_{t} )_{-n,k}
	\sqrt{k}\,
	h_{-k}
\Big\vert^{2}, \\
I\!\!I_{2}
&=
\sum_{n=1}^{\infty}
\Big\vert
	\sum_{k=1}^{\infty}
	( B_{s}^{*} ( A_{t} - A_{s} ) )_{-n,k}
	\sqrt{k}\,
	h_{-k}
\Big\vert^{2} .
\end{split}
\end{equation*}

By
Corollaries~\ref{B-bounded} and \ref{modulus-B-estimate},
we have
\begin{equation*}
\begin{split}
I\!\!I_{1}
&\leqslant 
\Vert
	B_{t}^{*} - B_{s}^{*}
\Vert^{2}
\Vert
	A_{t}
\Vert^{2}
\Vert
\Lambda
h
\Vert_{H}^{2}
\leqslant 
c_{21}
\omega (s,t)^{
2
}
\Vert h \Vert_{
H^{1/2}
}^{2}, \\
I\!\!I_{2}
&\leqslant 
\Vert
	B_{s}^{*}
\Vert^{2}
\Vert
	A_{t} - A_{s}
\Vert^{2}
\Vert
\Lambda
h
\Vert_{H}^{2}
\leqslant 
c_{22}
\omega (s,t)^{
2
}
\Vert h \Vert_{
H^{1/2}
}^{2},
\end{split}
\end{equation*}
for some constant $c_{21}, c_{22} > 0$.
Therefore we have obtained
\begin{equation*}
I\!\!I
\leqslant
c_{2}
\omega (s,t)^{
2
}
\Vert h \Vert_{
H^{1/2}
}^{2}
\end{equation*}
for some $c_{2} > 0$.

\noindent
\underline{{\it Estimate of $I\!\!I\!\!I$.}}

\begin{equation*}
\begin{split}
I\!\!I\!\!I
&=
\sum_{n=1}^{\infty}
\Big\vert
	\sum_{k=1}^{\infty}
	( B_{t} - B_{s} )_{n,-k}
	\sqrt{k}\,
	h_{k}
\Big\vert^{2}
\end{split}
\end{equation*}
is estimated by using
Corollary \ref{modulus-B-estimate}
as
\begin{equation*}
\begin{split}
I\!\!I\!\!I
&\leqslant
\Vert
	B_{t} - B_{s}
\Vert^{2}
\Vert
\Lambda
h
\Vert_{H}^{2}
\leqslant
c_{3} \omega (s,t)^{2}
\Vert h \Vert_{
H^{1/2}
}^{2}
\end{split}
\end{equation*}
for some constant $c_{3} > 0$.

\noindent
\underline{{\it Estimate of $I\!V$.}}

\begin{equation*}
\begin{split}
I\!V
&=
\sum_{n=1}^{\infty}
\Big\vert
	\sum_{k=1}^{\infty}
	(
	B_{t}B_{t}^{*} A_{t} B_{t}
	-
	B_{s}B_{s}^{*} A_{s} B_{s}
	)_{n,-k}
	\sqrt{k}\,
	h_{k}
\Big\vert^{2}
\end{split}
\end{equation*}

Along the decomposition
\begin{equation*}
\begin{split}
&
B_{t}B_{t}^{*} A_{t} B_{t}
-
B_{s}B_{s}^{*} A_{s} B_{s} \\
&=
( B_{t} - B_{s} ) B_{t}^{*} A_{t} B_{t}
+ B_{s} ( B_{t}^{*} A_{t} B_{t} - B_{s}^{*} A_{s} B_{s} ) \\
&=
( B_{t} - B_{s} ) B_{t}^{*} A_{t} B_{t}
+ B_{s} ( B_{t}^{*} - B_{s}^{*} ) A_{t} B_{t}
+ B_{s} B_{s}^{*} A_{s} B_{s} ( B_{s}^{*} - B_{t}^{*} ) A_{t} B_{t} \\
&\hspace{20mm}
+ B_{s} B_{s}^{*} A_{s} ( B_{s} - B_{t} ) B_{t}^{*} A_{t} B_{t}
+ B_{s} B_{s}^{*} A_{s} ( B_{t} - B_{s} ) ,
\end{split}
\end{equation*}
the quantity $I\!V$ is estimated as
$$
I\!V
\leqslant
5
(
	I\!V_{1}
	+ I\!V_{2}
	+ I\!V_{3}
	+ I\!V_{4}
	+ I\!V_{5}
),
$$
where
\begin{equation*}
\begin{split}
I\!V_{1}
&=
\sum_{n=1}^{\infty}
\Big\vert
	\sum_{k=1}^{\infty}
	(
	( B_{t} - B_{s} ) B_{t}^{*} A_{t} B_{t}
	)_{n,-k}
	\sqrt{k}\,
	h_{k}
\Big\vert^{2} , \\
I\!V_{2}
&=
\sum_{n=1}^{\infty}
\Big\vert
	\sum_{k=1}^{\infty}
	(
	B_{s} ( B_{t}^{*} - B_{s}^{*} ) A_{t} B_{t}
	)_{n,-k}
	\sqrt{k}\,
	h_{k}
\Big\vert^{2} , \\
I\!V_{3}
&=
\sum_{n=1}^{\infty}
\Big\vert
	\sum_{k=1}^{\infty}
	(
	B_{s} B_{s}^{*} A_{s} B_{s} ( B_{s}^{*} - B_{t}^{*} ) A_{t} B_{t}
	)_{n,-k}
	\sqrt{k}\,
	h_{k}
\Big\vert^{2} , \\
I\!V_{4}
&=
\sum_{n=1}^{\infty}
\Big\vert
	\sum_{k=1}^{\infty}
	(
	B_{s} B_{s}^{*} A_{s} ( B_{s} - B_{t} ) B_{t}^{*} A_{t} B_{t}
	)_{n,-k}
	\sqrt{k}\,
	h_{k}
\Big\vert^{2} , \\
I\!V_{5}
&=
\sum_{n=1}^{\infty}
\Big\vert
	\sum_{k=1}^{\infty}
	(
	B_{s} B_{s}^{*} A_{s} ( B_{t} - B_{s} )
	)_{n,-k}
	\sqrt{k}\,
	h_{k}
\Big\vert^{2} .
\end{split}
\end{equation*}
By using
Corollaries~\ref{B-bounded} and \ref{modulus-B-estimate},
it is easy to see that
\begin{equation*}
I\!V_{i}
\leqslant
c_{5i} \omega (s,t)^{2}
\Vert h \Vert_{
	H^{1/2}
}^{2}
\quad
\text{for $i=1,2,3,4,5,$}
\end{equation*}
for some constants
$
c_{51}, c_{52}, c_{53}, c_{54}, c_{55} > 0
$.
Therefore we get
\begin{equation*}
I\!V
\leqslant
c_{5} \omega (s,t)^{2}
\Vert h \Vert_{
	H^{1/2}
}^{2}
\end{equation*}
for some constant $c_{5} > 0$.

\noindent
\underline{{\it Estimate of $V$.}}

\begin{equation*}
\begin{split}
V
&=
\sum_{n=1}^{\infty}
\Big\vert
	\sum_{k=1}^{\infty}
	( B_{t} B_{t}^{*} A_{t} - B_{s} B_{s}^{*} A_{s} )_{n,k}
	\sqrt{k}\,
	h_{-k}
\Big\vert^{2}
\end{split}
\end{equation*}

Along the decomposition
\begin{equation*}
\begin{split}
&
B_{t} B_{t}^{*} A_{t} - B_{s} B_{s}^{*} A_{s} \\
&=
( B_{t} - B_{s} ) B_{t}^{*} A_{t}
+ B_{s} ( B_{t}^{*} - B_{s}^{*} ) A_{t}
+ B_{s} B_{s}^{*} ( A_{t} - A_{s} ) ,
\end{split}
\end{equation*}
the quantity $V$ is estimated as
$$
V \leqslant 3 ( V_{1} + V_{2} + V_{3} ),
$$
where
\begin{equation*}
\begin{split}
V_{1}
&=
\sum_{n=1}^{\infty}
\Big\vert
	\sum_{k=1}^{\infty}
	(
	( B_{t} - B_{s} ) B_{t}^{*} A_{t} )_{n,k}
	\sqrt{k}\,
	h_{-k}
\Big\vert^{2} \\
V_{2}
&=
\sum_{n=1}^{\infty}
\Big\vert
	\sum_{k=1}^{\infty}
	(
	B_{s} ( B_{t}^{*} - B_{s}^{*} ) A_{t} )_{n,k}
	\sqrt{k}\,
	h_{-k}
\Big\vert^{2}, \\
V_{3}
&=
\sum_{n=1}^{\infty}
\Big\vert
	\sum_{k=1}^{\infty}
	(
	B_{s} B_{s}^{*} ( A_{t} - A_{s} ) )_{n,k}
	\sqrt{k}\,
	h_{-k}
\Big\vert^{2} .
\end{split}
\end{equation*}
By
Corollaries~\ref{B-bounded} and \ref{modulus-B-estimate},
we have
\begin{equation*}
\begin{split}
V_{1}
&\leqslant
\Vert
B_{t} - B_{s}
\Vert^{2}
\Vert
B_{t}^{*} A_{t}
\Vert^{2}
\Vert
	\Lambda
h
\Vert_{H}^{2}
\leqslant
c_{51} \omega (s,t)^{2}
\Vert h \Vert_{
	H^{1/2}
}^{2}, \\
V_{2}
&\leqslant
\Vert
B_{s}
\Vert^{2}
\Vert
B_{t}^{*} - B_{s}^{*}
\Vert^{2}
\Vert
A_{t}
\Vert^{2}
\Vert
	\Lambda
h
\Vert_{H}^{2}
\leqslant
c_{52}
\omega (s,t)^{
2
}
\Vert h \Vert_{
	H^{1/2}
}^{2} , \\
V_{3}
&\leqslant
\Vert
B_{s} B_{s}^{*} \Vert^{2}
\Vert
A_{t} - A_{s}
\Vert^{2}
\Vert
	\Lambda
h
\Vert_{H}^{2}
\leqslant
c_{53}
\omega (s,t)^{
2
}
\Vert h \Vert_{
	H^{1/2}
}^{2} ,
\end{split}
\end{equation*}
for some constants $c_{51}, c_{52}, c_{53} > 0$.
Combining these estimates, we get
\begin{equation*}
\begin{split}
V
\leqslant
c_{5}
\omega (s,t)^{
2
}
\Vert h \Vert_{
	H^{1/2}
}^{2} .
\end{split}
\end{equation*}
for some constant $c_{5} > 0$.

Now by combining the estimates for
$I$, $I\!\!I$, $I\!\!I\!\!I$, $I\!V$ and $V$,
we obtain the assertion.
\end{proof} 

\appendix
\section{\ }
\label{Appdx} 

\subsection{Proof of Corollary~\ref{boundedness(sol)}}
\label{app:bdd-sol} 

We write $f_{t}(z)$ as
$
f_{t} (z)
=
C(t) \sum_{n=1}^{\infty} c_{n}(t) z^{n}
$,
and then it is enough to show that
$
\sup_{ z \in \mathbb{D} }
\sum_{n=1}^{\infty}
\vert c_{n} (t) z^{n} \vert
< +\infty,
$
for each $t \in [0,T]$.
By \cite[Theorem~2.8]{AmbFr},
\begin{equation*}
\begin{split}
\vert c_{n}(t) \vert
\leqslant
\sum_{p=1}^{n}
\sum_{\substack{
	i_{1}, \cdots , i_{p} \in \mathbb{N}: \\
	i_{1} + \cdots + i_{p} = n
}}
\widetilde{w} (n)_{ i_{1}, \cdots , i_{p} }
\frac{
	\omega (0,t)^{n}
}{
	n!
},
\end{split}
\end{equation*}
where
\begin{equation*}
\begin{split}
\widetilde{w} (n)_{ i_{1}, \cdots , i_{p} }
&=
\{ ( n - i_{1} ) + 1 \}
\{ ( n - ( i_{1} + i_{2} ) ) + 1 \}
\cdots
\{ ( n - ( i_{1} + \cdots + i_{p-1} ) ) + 1 \} \\
&\leqslant
\{ ( n - 1 ) + 1 \}
\{ ( n - 2 ) + 1 \}
\cdots
\{ ( n - ( p-1 ) ) + 1 \} \\
&=
n
(n-1)
\cdots
(n-p)
=
n \binom{n-1}{p} p!
\leqslant
n 2^{n-1} p! .
\end{split}
\end{equation*}
Therefore,
\begin{equation*}
\begin{split}
\vert c_{n}(t) \vert
& 
\leqslant
n 2^{n-1}
\frac{
	\omega (0,t)^{n}
}{
	n!
}
\sum_{p=1}^{n}
p!
\sum_{\substack{
	i_{1}, \cdots , i_{p} \in \mathbb{N}: \\
	i_{1} + \cdots + i_{p} = n
}}
1 \\
&= 
n 2^{n-1}
\frac{
	\omega (0,t)^{n}
}{
	n!
}
\sum_{p=1}^{n}
p!
\binom{n-1}{p-1} \\
& 
\leqslant
n 2^{n-1}
\omega (0,t)^{n}
\sum_{p=1}^{n}
\binom{n-1}{p-1} \\
&= 
n 4^{n-1}
\omega (0,t)^{n}
=
4^{-1}
n
( 4 \omega (0,t) )^{n} .
\end{split}
\end{equation*}
Hence,
$
\sup_{ z \in \mathbb{D} }
\sum_{n=1}^{\infty}
\vert c_{n} (t) z^{n} \vert
\leqslant
\sum_{n=1}^{\infty}
\vert c_{n} (t) \vert
< +\infty
$
if $\omega (0,T) < \frac{1}{4}$.

\subsection{Proof of Corollary~\ref{Loew-Grunsky-estimate}}
\label{app:Grunsky-estimate} 

(i)
is immediate from Definition~\ref{LKw/omega} since
$b_{-1,-1}(t)$ is explicitly given by
$$
b_{-1,-1} (t)
=
- \mathrm{e}^{2x_{0} (t)}
\int_{0}^{t} \mathrm{e}^{ -2x_{0} (u) } \mathrm{d} x_{2} (u)
$$
(see \cite[Proposition~2.10--(ii)]{AmbFr}).

(ii)
Since the Loewner-Kufarev equation is
controlled by $\omega$,
we have that
\begin{equation*}
\begin{split}
\mathrm{e}^{ (n+m) x_{0} (t) }
\Big\vert
\int_{0}^{t}
(
	( x_{i_{p}} \cdots x_{i_{2}} x_{i_{1}} )
	\shuffle
	( x_{j_{q}} \cdots x_{j_{2}} x_{j_{1}} )
)(u)
\mathrm{d} x_{k} (u)
\Big\vert
\leqslant
\frac{ (p+q)! }{ p! q! }
\frac{\omega (0,t)^{n+m}}{(n+m)!}
\end{split}
\end{equation*}
(for the notation $\shuffle$, see \cite[Definition~2.11]{AmbFr})
if
$
k + ( i_{1} + \cdots + i_{p} ) + (j_{1} + \cdots + j_{q})
=
n+m
$.
Then by
\cite[Theorem~2.12]{AmbFr},
we get
\begin{equation*}
\begin{split}
&
\vert
	b_{-m,-n} (t)
\vert \\
&\leqslant
\sum_{k=2}^{n+m-2}
\sum_{\substack{
	1 \leqslant i < m; \\
	1 \leqslant j < n: \\
	i+j=k
}}
\sum_{p=1}^{m-i}
\sum_{q=1}^{n-j}
\sum_{\substack{
	i_{1}, \cdots , i_{p} \in \mathbb{N}: \\
	i_{1} + \cdots + i_{p} = m-i
}}
\sum_{\substack{
	j_{1}, \cdots , j_{q} \in \mathbb{N}: \\
	j_{1} + \cdots + j_{q} = n-j
}}
w_{i_{1}, \cdots , i_{p}; j_{1}, \cdots , j_{q}}
\frac{ (p+q)! }{ p! q! }
\frac{\omega (0,t)^{n+m}}{(n+m)!} \\
&\hspace{10mm}+
\sum_{k=m+1}^{n+m-1}
\sum_{q=1}^{n+m-k}
\sum_{\substack{
	j_{1}, \cdots , j_{q} \in \mathbb{N}: \\
	j_{1} + \cdots + j_{q} = n+m-k
}}
w_{ \varnothing ; j_{1}, \cdots , j_{q} }
\frac{ \omega (0,t)^{n+m} }{ (n+m)! } \\
&\hspace{20mm}+
\sum_{k=n+1}^{n+m-1}
\sum_{p=1}^{m+n-k}
\sum_{\substack{
	i_{1}, \cdots , i_{p} \in \mathbb{N}: \\
	i_{1} + \cdots + i_{p} = m+n-k
}}
w_{ i_{1}, \cdots , i_{p}; \varnothing }
\frac{ \omega (0,t)^{n+m} }{ (n+m)! },
\end{split}
\end{equation*}
where,
in the first term on the right hand side,
we have
$( m - ( i_{1} + \cdots + i_{p} )) = i$,
$( n - ( j_{1} + \cdots + j_{q} )) = j$
and hence
\begin{equation*}
\begin{split}
w_{i_{1}, \cdots , i_{p}; j_{1}, \cdots , j_{q}}
&=
( m-i_{1} ) ( m- (i_{1}+i_{2}) ) \cdots
( m - (i_{1}+\cdots + i_{p}) ) \\
&\hspace{10mm}\times
( n-j_{1} ) ( n- (j_{1}+j_{2}) ) \cdots
( n - (j_{1}+\cdots + j_{q}) ) \\
&\leqslant 
i (m-1) (m-2) \cdots (m-(p-1)) \\
&\hspace{10mm}\times
j (n-1) (n-2) \cdots (n-(q-1)) \\
&= 
ij \binom{m-1}{p-1} (p-1)! \binom{n-1}{q-1} (q-1)! .
\end{split}
\end{equation*}
In the second and third term, we have
$( n - ( j_{1} + \cdots + j_{q} )) = k-m$,
$( m - ( i_{1} + \cdots + i_{p} )) = k-n$,
so that
\begin{equation*}
\begin{split}
&
w_{ \varnothing ; j_{1}, \cdots , j_{q} }
\leqslant
(k-m)
\binom{ n-1 }{ q-1 }
(q-1)! , \\
& 
w_{ i_{1}, \cdots , i_{q} ; \varnothing }
\leqslant
(k-n)
\binom{ m-1 }{ p-1 }
(p-1)! .
\end{split}
\end{equation*}
So we have
\begin{equation*}
\begin{split}
&
\vert
	b_{-m,-n} (t)
\vert \\
&\leqslant 
\sum_{k=2}^{n+m-2}
\sum_{\substack{
	1 \leqslant i < m; \\
	1 \leqslant j < n: \\
	i+j=k
}}
\sum_{p=1}^{m-i}
\sum_{q=1}^{n-j}
\sum_{\substack{
	i_{1}, \cdots , i_{p} \in \mathbb{N}: \\
	i_{1} + \cdots + i_{p} = m-i
}}
\sum_{\substack{
	j_{1}, \cdots , j_{q} \in \mathbb{N}: \\
	j_{1} + \cdots + j_{q} = n-j
}}
ij
\binom{m-1}{p-1}
\binom{n-1}{q-1}
(p-1)! (q-1)! \\
&\hspace{40mm}\times
\frac{ (p+q)! }{ p! q! }
\frac{\omega (0,t)^{n+m}}{(n+m)!} \\
&\hspace{10mm}+
\sum_{k=m+1}^{n+m-1}
\sum_{q=1}^{n+m-k}
\sum_{\substack{
	j_{1}, \cdots , j_{q} \in \mathbb{N}: \\
	j_{1} + \cdots + j_{q} = n+m-k
}}
(k-m)
\binom{ n-1 }{ q-1 }
(q-1)!
\frac{ \omega (0,t)^{n+m} }{ (n+m)! } \\
&\hspace{20mm}+
\sum_{k=n+1}^{n+m-1}
\sum_{p=1}^{m+n-k}
\sum_{\substack{
	i_{1}, \cdots , i_{p} \in \mathbb{N}: \\
	i_{1} + \cdots + i_{p} = m+n-k
}}
(k-n)
\binom{ m-1 }{ p-1 }
(p-1)!
\frac{ \omega (0,t)^{n+m} }{ (n+m)! } \\
&= 
\sum_{k=2}^{n+m-2}
\sum_{\substack{
	1 \leqslant i < m; \\
	1 \leqslant j < n: \\
	i+j=k
}}
ij
\sum_{p=1}^{m-i}
\sum_{q=1}^{n-j}
\Big(
\sum_{\substack{
	i_{1}, \cdots , i_{p} \in \mathbb{N}: \\
	i_{1} + \cdots + i_{p} = m-i
}}
1
\Big)
\Big(
\sum_{\substack{
	j_{1}, \cdots , j_{q} \in \mathbb{N}: \\
	j_{1} + \cdots + j_{q} = n-j
}}
1
\Big) \\
&\hspace{15mm}\times
\binom{m-1}{p-1}
\binom{n-1}{q-1}
(p-1)! (q-1)!
\frac{ (p+q)! }{ p! q! }
\frac{\omega (0,t)^{n+m}}{(n+m)!} \\
&\hspace{10mm}+
\sum_{k=m+1}^{n+m-1}
\sum_{q=1}^{n+m-k}
\Big(
\sum_{\substack{
	j_{1}, \cdots , j_{q} \in \mathbb{N}: \\
	j_{1} + \cdots + j_{q} = n+m-k
}}
1
\Big)
(k-m)
\binom{ n-1 }{ q-1 }
(q-1)!
\frac{ \omega (0,t)^{n+m} }{ (n+m)! } \\
&\hspace{20mm}+
\sum_{k=n+1}^{n+m-1}
\sum_{p=1}^{m+n-k}
\Big(
\sum_{\substack{
	i_{1}, \cdots , i_{p} \in \mathbb{N}: \\
	i_{1} + \cdots + i_{p} = m+n-k
}}
1
\Big)
(k-n)
\binom{ m-1 }{ p-1 }
(p-1)!
\frac{ \omega (0,t)^{n+m} }{ (n+m)! }
\end{split}
\end{equation*}
In the first term of the last equation,
we have
$
\frac{ (p+q)! }{ p! q! }
\leqslant
2^{p+q} \leqslant 2^{m+n-2}
$,
$
(p-1)! (q-1)!
\leqslant (p+q-1)!
\leqslant (m+n-3)!
$,
and
$$
\sum_{\substack{
	i_{1}, \cdots , i_{p} \in \mathbb{N}: \\
	i_{1} + \cdots + i_{p} = m-i
}}
1
=
\binom{ m-i-1 }{ p-1 } ,
\quad
\sum_{\substack{
	j_{1}, \cdots , j_{q} \in \mathbb{N}: \\
	j_{1} + \cdots + j_{q} = n-j
}}
1
= \binom{ n-j-1 }{ q-1 }.
$$
For the second and the third terms, we have
$
(q-1)! \leqslant (n+m-3)!
$
and
$
(p-1)! \leqslant (n+m-3)!
$
respectively.
Then we obtain
\begin{equation*}
\begin{split}
&
\vert
	b_{-m,-n} (t)
\vert \\
&\leqslant 
\omega (0,t)^{n+m}
\frac{ (m+n-3)! 2^{m+n-2} }{ (n+m)! }
\Big\{
\sum_{k=2}^{n+m-2}
\sum_{\substack{
	1 \leqslant i < m; \\
	1 \leqslant j < n: \\
	i+j=k
}}
ij \\
&\hspace{5mm}\times
\Big(
\sum_{p^{\prime}=0}^{m-i-1}
\binom{ m-1 }{ p^{\prime} }
\binom{ m-i-1 }{ (m-i-1) - p^{\prime} }
\Big)
\Big(
\sum_{ q^{\prime} = 0 }^{ n-j-1 }
\binom{ n-1 }{ q^{\prime} }
\binom{ n-j-1 }{ (n-j-1) - q^{\prime} }
\Big) \\
&\hspace{15mm}+
\sum_{k=m+1}^{n+m-1}
(k-m)
\sum_{q=1}^{n+m-k}
\binom{ n+m-k-1 }{ q-1 }
\binom{ n-1 }{ q-1 } \\
&\hspace{25mm}+
\sum_{k=n+1}^{n+m-1}
(k-n)
\sum_{p=1}^{m+n-k}
\binom{ n+m-k-1 }{ p-1 }
\binom{ m-1 }{ p-1 }
\Big\} .
\end{split}
\end{equation*}
The combinatorial formulae
\begin{equation*}
\begin{split}
\sum_{p^{\prime}=0}^{m-i-1}
\binom{ m-1 }{ p^{\prime} }
\binom{ m-i-1 }{ (m-i-1) - p^{\prime} }
&=
\binom{ 2(m-1)-i }{ m-1 } , \\
\sum_{ q^{\prime} = 0 }^{n-j-1}
\binom{ n-1 }{ q^{\prime} }
\binom{ n-j-1 }{ (n-j-1) - q^{\prime} }
&=
\binom{ 2(n-1)-j }{ n-1 }
\end{split}
\end{equation*}
give
\begin{equation*}
\begin{split}
&
\vert
	b_{-m,-n} (t)
\vert \\
&\leqslant 
\omega (0,t)^{n+m}
\frac{ (m+n-3)! 2^{m+n-2} }{ (n+m)! }
\Big\{
\sum_{k=2}^{n+m-2}
\sum_{\substack{
	1 \leqslant i < m; \\
	1 \leqslant j < n: \\
	i+j=k
}}
ij
\binom{ 2(m-1)-i }{ m-1 }
\binom{ 2(n-1)-j }{ n-1 } \\
&\hspace{15mm}+
\sum_{k=m+1}^{n+m-1}
(k-m)
\binom{ 2(n-1) - (k-m) }{ n-1 } \\
&\hspace{25mm}+
\sum_{k=n+1}^{n+m-1}
(k-n)
\binom{ 2(m-1) - (k-n) }{ m-1 }
\Big\} \\
&= 
\omega (0,t)^{n+m}
\frac{ (m+n-3)! 2^{m+n-2} }{ (n+m)! }
\Big\{
\Big(
\sum_{ i = 1 }^{ m-1 }
i
\binom{ 2(m-1)-i }{ m-1 }
\Big)
\Big(
\sum_{ j = 1 }^{ n-1 }
j
\binom{ 2(n-1)-j }{ n-1 }
\Big) \\
&\hspace{15mm}+
\sum_{k=1}^{n-1}
k
\binom{ 2(n-1) - k }{ n-1 }
+
\sum_{k=1}^{m-1}
(k-n)
\binom{ 2(m-1) - k }{ m-1 }
\Big\} .
\end{split}
\end{equation*}
We shall use here the identity
$
\sum_{j=k}^{n}
( n+1-j )
\binom{ j-1 }{ k-1 }
=
\binom{ n+1 }{ k+1 }
$
which enables us to conclude
\begin{equation*}
\begin{split}
\sum_{ i = 1 }^{ m-1 }
i
\binom{ 2(m-1)-i }{ m-1 }
=
\sum_{j=m}^{2(m-1)}
( 2(m-1) + 1 - j )
\binom{ j-1 }{ m-1 }
=
\binom{2m-1}{m}.
\end{split}
\end{equation*}

Hence, we conclude that
\begin{equation*}
\begin{split}
&
\vert
	b_{-m,-n} (t)
\vert \\
&\leqslant 
\omega (0,t)^{m+n}
\frac{ (m+n-3)! 2^{m+n-2} }{ (m+n)! }
\Big\{
\binom{2m-1}{m}
\binom{2n-1}{n}
+
\binom{2n-1}{n}
+
\binom{2m-1}{m}
\Big\} \\
&\leqslant 
\omega (0,t)^{m+n}
\frac{ (m+n-3)! 2^{m+n-2} }{ (m+n)! }
\Big\{
\binom{2m-1}{m} + 1
\Big\}
\Big\{
\binom{2n-1}{n} + 1
\Big\} \\
&\leqslant 
\frac{
	\omega (0,t)^{m+n}
	2^{m+n-2}
	( 2^{2m-1} 2^{2n-1} )
}{
	(m+n) (m+n-1) (m+n-2)
}
= 
\frac{
	(8\omega (0,t))^{m+n}
}{
	16 (m+n) (m+n-1) (m+n-2)
}.
\end{split}
\end{equation*}

(iii)
Since the Loewner-Kufarev equation is
controlled by $\omega$,
we have
\begin{equation*}
\begin{split}
&
\Big\vert
\mathrm{e}^{ (n+m) x_{0}(t) }
\int_{0}^{t}
\mathrm{d}
(
	( x_{i_{p}} \cdots x_{i_{2}} x_{i_{1}} )
	\shuffle
	( x_{j_{q}} \cdots x_{j_{2}} x_{j_{1}} )
)(u)
\int_{0}^{u}
\mathrm{e}^{-kx_{0} (v)}
\mathrm{d} x_{k} (v) \\
&\hspace{10mm} -
\mathrm{e}^{ (n+m) x_{0}(s) }
\int_{0}^{s}
\mathrm{d}
(
	( x_{i_{p}} \cdots x_{i_{2}} x_{i_{1}} )
	\shuffle
	( x_{j_{q}} \cdots x_{j_{2}} x_{j_{1}} )
)(u)
\int_{0}^{u}
\mathrm{e}^{-kx_{0} (v)}
\mathrm{d} x_{k} (v)
\Big\vert \\
&\hspace{10mm}\leqslant 
(n+m)
\frac{ (p+q)! }{ p! q! }
\frac{
	\omega (s,t) \omega (0,T)^{n+m-1}
}{
	(n+m)!
} .
\end{split}
\end{equation*}
Now the remaining case is the same as (i).

\subsection*{Acknowledgment}
The first author was supported by JSPS KAKENHI Grant Number 15K17562.
The second author was
partially supported
by the ERC advanced grant
``Noncommutative distributions in free probability".
Both authors thank Theo Sturm for the hospitality he granted to T.A. at the University of Bonn. T.A. thanks Roland Speicher for the hospitality offered him in Saarbr\"ucken.  R.F. thanks Fukuoka University and the MPI in Bonn for its hospitality, Roland Speicher for his support from which this paper substantially profited.

\end{document}